\documentclass[12pt,onecolumn]{IEEEtran}

\usepackage{verbatim}
\usepackage{epsfig,bbm}
\usepackage[colorlinks,bookmarksopen,bookmarksnumbered,citecolor=red,urlcolor=red,]{hyperref}
\usepackage{CJK}
\usepackage{indentfirst}
\usepackage{multirow}
\usepackage{epstopdf}
\usepackage{graphicx}
\usepackage{footmisc}
\graphicspath{{./figures/}}
\usepackage{amsfonts}
\usepackage{mathrsfs}
\usepackage{setspace}
\usepackage{amsmath}
\usepackage{algorithm,algorithmic,amsbsy,amsmath,amssymb,epsfig,bbm,mathrsfs, bbm} 
\usepackage{amsthm}
\usepackage{verbatim}
\usepackage[noadjust]{cite}
\usepackage[subfigure]{graphfig}

\newtheorem{theorem}{Theorem}

\newtheorem{corollary}{Corollary}
\newtheorem{prop}{Proposition}

\usepackage{xcolor}
\allowdisplaybreaks

\begin{document}

\title{
Wireless-Powered Device-to-Device Communications with Ambient Backscattering: Performance Modeling and Analysis
 }
 
%
\author{Xiao Lu, Hai Jiang~\emph{Senior Member, IEEE}, Dusit Niyato~\emph{Fellow, IEEE}, Dong In Kim~\emph{Senior Member, IEEE}, and Zhu Han~\emph{Fellow, IEEE}   \\
\thanks{ \hspace{-4mm}  
This work was supported in part by the National Research Foundation of Korea Grant funded by the Korean Government under Grant 2014R1A5A1011478, the Natural Sciences and Engineering
Research Council of Canada, Singapore MOE Tier 2 under Grant MOE2014-T2-2-015 ARC4/15 and NRF2015-NRF-ISF001-2277, and US NSF CNS-1717454, CNS-1731424, CNS-1702850, CNS-1646607,ECCS-1547201,CMMI-1434789,CNS-1443917, and ECCS-1405121. (Corresponding author: Dong In Kim)
   
Xiao Lu and Hai Jiang are with the Dept. of Electrical \& Computer Engineering, University of Alberta, Edmonton, AB T6G 1H9, Canada (email: lu9@ualberta.ca and hai1@ualberta.ca).

Dusit Niyato is with the School of Computer Science $\&$ Engineering, Nanyang Technological University (NTU), Singapore (email: dniyato@ntu.edu.sg). 

Dong In Kim is with the School of Information and Communication Engineering, Sungkyunkwan University (SKKU), South Korea (email: dikim@skku.ac.kr).  

Zhu Han is with the University of Houston, Houston, TX 77004 USA (e-mail:zhan2@uh.edu), and also with the Department of Computer Science and Engineering, Kyung Hee University, Seoul, South Korea (email: zhan2@uh.edu). 

} 
}

\markboth{}{Shell \MakeLowercase{\textit{et al.}}: Bare Demo of
IEEEtran.cls for Journals}

\maketitle 

\begin{abstract} 

The recent advanced 
wireless energy harvesting technology has enabled wireless-powered communications  
to accommodate 
wireless data services in a self-sustainable manner. However, wireless-powered communications rely on active RF signals to communicate, and result in high power consumption.
On the other hand, ambient backscatter technology that passively reflects existing RF signal sources in the air to communicate 
has the potential to facilitate an implementation with ultra-low power consumption.  
In this paper, we introduce a hybrid D2D communication paradigm by integrating ambient backscattering with wireless-powered communications. 
The hybrid D2D communications are self-sustainable, as no dedicated external power supply is required.
However, since the radio signals for energy harvesting and for backscattering come from the ambient, the performance of the hybrid D2D communications depends largely on environment factors, e.g., distribution, spatial density, and transmission load of ambient energy sources. Therefore, we design two mode selection protocols for the hybrid D2D transmitter, allowing a more flexible adaptation to the environment. We then introduce analytical models to characterize the impacts of the considered environment factors on the hybrid D2D communication performance. Together with extensive simulations, our analysis shows that the communication performance benefits from larger repulsion, transmission load and density of ambient energy sources. Further, we investigate how different mode selection mechanisms affect the communication performance. 

%






%
 
\end{abstract}

\emph{Index terms- Ambient backscatter, D2D communications, modulated backscatter, repulsion behavior wireless-powered communications}.   
\section{Introduction}
The Internet-of-Things (IoT)~\cite{V.2016Gazis
,D.August2016Niyato} aims to connect a large number of  intelligent devices, such as smart household devices~\cite{D.April_2011Niyato}, renewable sensors~\cite{D.May_2016Niyato}, vehicular communicators, RFID tags, and wearable health-care gadgets. 
Many such devices with small size and simple implementation can only achieve short-range and low-rate communications. In this context, device-to-device (D2D) communications~\cite{Camps-Mur2013}, which empower two devices in proximity to establish direct connections, 
appear as a cost-effective and energy-efficient solution to accommodate ubiquitous short-range connections. Recent research efforts~\cite{J.2015Liu,A.2015Al-Fuqaha ,X.March_2015Lu} have shown that D2D communications can achieve significant performance gains in terms of network coverage, capacity, peak rate and communication latency. Thus, D2D communications are considered as an intrinsic part of the IoT. 
 


Recently, radio frequency (RF) energy harvesting~\cite{XLuSurvey ,X.2014Lu} has evolved as a promising energy replenishment solution to empower D2D communications with self-sustainability~\cite{H.2015Sakr,H.2016Yang,Y.2016Liu,X.Lu2016}. 
Technically, D2D transmitters can harvest energy from RF signals in the air for their operation. By utilizing the harvested energy, the D2D transmitters can communicate without relying on an external energy supply, which will technically and economically facilitate a massive deployment of D2D communication devices.  
Thus, a zero-energy communication paradigm can be virtually achieved. 

%


The recent advance in RF energy harvesting technologies has paved the way for two emerging green communication technologies, namely, wireless-powered communications \cite{XLuSurvey,H.2014Ju,DJune_2015Niyato,DMarch_2015Niyato} and ambient backscatter communications~\cite{V.2013Liu,G.2016Wang}. 



\begin{itemize}

\item  
Wireless-powered communications utilize the harvested energy to generate RF signals for information transmission. The energy harvesting can be performed either opportunistically from ambient RF signal sources (e.g., cellular base stations and mobile terminals)~\cite{I.2015Flint} or in a more controlled way from dedicated RF power beacons~\cite{K.2014Huang}. 
A comprehensive survey of wireless-powered communications can be found in \cite{X.Sec2016Lu}.
 
 


\item  Ambient backscattering performs data transmission based on modulated backscatter of existing RF signals,
e.g., from TV base stations~\cite{V.2013Liu} or Wi-Fi routers~\cite{D.2015Bharadia}. 
Importantly, in contrast to traditional backscatter devices, e.g., RFID tags, that passively rely on a dedicated interrogator to initiate communication, an ambient backscatter transmitter can actively initiate communication to its peers. 

\end{itemize}
 








Despite many benefits, both wireless-power communications and ambient backscatter communications have drawbacks that limit their applicability for D2D communications. Specifically, a wireless-powered device can only communicate intermittently as it requires dedicated time for energy harvesting. To perform active RF generation, the required energy is much higher than the operation power of ambient backscattering.  
As for ambient backscatter communications, the relatively low data rate, typically ranging from several to tens of kbps~\cite{V.2013Liu,B.2014Kellogg}, largely constrains the applications. 
A relatively high signal-to-noise ratio (SNR) is required to realize reliable transmission with modulated backscatter. 
Moreover, the transmission distance is limited, typically ranging from several feet to tens of feet \cite{V.2013Liu,N.2014Parks}. 
To address these shortcomings, in this paper, we introduce a novel hybrid D2D communication paradigm that harvests energy from ambient RF signals and can selectively perform ambient backscattering or wireless-powered communications for improved applicability and functionality.
Through analysis, we will show that these two technologies can well complement each other and result in better performance for D2D communications.

\subsection{Related Work}

Recently, wireless-powered communications~\cite{XLuSurvey} have attracted much attention and have been applied in D2D communications to improve energy efficiency. 
In~\cite{H.2015Sakr}, the authors investigate a cognitive D2D transmitter that harvests energy from cellular users and transmits on a selected cellular channel. The study focuses on the impact of different spectrum access schemes on the transmission outage probability. In~\cite{H.2016Yang}, the authors propose a selection scheme for cellular users to choose between wireless-powered D2D relaying and direct transmission. Under a K-tier heterogeneous network model, the outage probability of cellular users is derived in closed-form expressions. Both~\cite{H.2015Sakr} and~\cite{H.2016Yang} aim at improving the self-sustainable D2D communications through efficient spectrum allocation. However, the use of wireless-powered transmission is subject to channel availability. In particular, wireless-powered transmission is not feasible when all the channels are occupied. 
Different from these research efforts, our hybrid design allows transmission by ambient backscattering when wireless-powered transmission is infeasible, which does not cause noticeable interference to legitimate users~\cite{V.2013Liu,D.2016Darsena}.

Different from the works~\cite{H.2015Sakr,H.2016Yang} that consider ambient RF energy harvesting, the work in~\cite{Y.2016Liu} studies D2D communications with dedicated power beacons for wireless energy provisioning. Both energy outage and secrecy outage probabilities are analyzed under different power beacon allocation schemes. The work in~\cite{K.2017Han} introduces a system model termed wirelessly powered backscatter communication network, which utilizes
dedicated power beacons transmitting unmodulated carrier signals to power the network nodes. 
Once successfully powered, each node can transmit information by backscattering the signals from the same power beacon. 
From the studies in~\cite{Y.2016Liu} and \cite{K.2017Han}, it is confirmed that adopting power beacons increases the available wireless power, and thus facilitates both wireless-powered communications and backscatter communications. However, this approach is costly and not energy-efficient due to the use of power beacons. 
 
More recently, ambient backscatter communications have been analyzed in wireless network environments. 
The authors in \cite{T.2017Hoang} 
investigate a cognitive radio network where a wireless-powered secondary user can either harvest energy or adopt ambient backscattering from a primary user on transmission. 
To maximize the throughput of the secondary user, a time allocation problem is developed to obtain the optimal time ratio between energy harvesting and ambient backscattering. The work in \cite{H.2017Kim} introduces a hybrid backscatter communication as an alternative access scheme for a wireless-powered transmitter. Specifically, when the ambient RF signals are not sufficient to support wireless-powered communications, the transmitter can adopt either bistatic backscattering or ambient backscattering depending on the availability of a dedicated carrier emitter.  A throughput maximization problem is formulated to find the optimal time allocation for the hybrid backscatter communication operation. Both \cite{T.2017Hoang} and \cite{H.2017Kim} study deterministic scenarios. 
Instead, our work takes into account the spatial randomness of network components and focuses on investigating the impact of  different spatial distributions.
 






\subsection{Motivation and Contributions} 
For sustainable D2D communications, the energy harvested by the D2D transmitter and the interference that impairs the D2D communications both come from RF signals emitted by ambient transmitters (e.g., cellular mobiles). These RF signals are dependent on environment factors, such as distribution, transmission load, and density of ambient transmitters. 
Such a dual nature of electromagnetic interference has stimulated the upsurge of interest for communication
systems powered by ambient energy harvested (see \cite{XLuSurvey,S.2015Bi} and references therein).

Unlike conventional wireless communications with reliable and stable power supply, the sustainability of our proposed hybrid D2D communications depends on the stochastic nature of wireless channels. Moreover, 
environment factors, e.g., distribution, spatial density, and transmission load of ambient energy sources, significantly affect the performance of the proposed hybrid D2D communications. To understand the performance of the hybrid communications, it is important to study the role of ambient RF signals which serve as energy resources for harvesting, signal resources for ambient backscattering, and interferers that affect the D2D transmission. This motivates us to investigate and characterize the impact of environment factors on the self-sustainable D2D communications.

The main contributions of our work are summarized as follows.

\begin{itemize}

\item First, we propose a novel self-sustainable communication paradigm for D2D networks, namely, ambient backscattering assisted wireless-powered communications. In particular, a hybrid D2D system that combines both ambient backscattering and wireless-powered communication capabilities is introduced. 

\item Taking into account environment factors, we analyze the hybrid D2D communications in a large-scale wireless communication network.

\item We propose two mode selection protocols for the hybrid D2D communications. For each protocol, we theoretically characterize the energy outage probability, coverage probability (i.e., successful transmission probability), and average throughput of the hybrid D2D communications.


\item We validate the analysis by simulations and investigate the performance of the hybrid D2D communications under various conditions. The evaluation shows that the hybrid transmitter is more flexible than a pure wireless-powered transmitter and a pure backscatter transmitter. Moreover, the hybrid D2D communications benefit from larger repulsion degree, transmission load and spatial density of ambient energy sources. 

\end{itemize}

The rest of this paper is structured as follows. Section II introduces the design of ambient backscattering assisted wireless-powered communications. Section III presents the system model along with the geometric network modeling. Section IV then theoretically characterizes the performance of the hybrid communication paradigm with regard to different metrics.  
Section V presents the validation based on Monte Carlo simulations and discussions. Finally, Section VI concludes our work and indicates future research directions. 

{\bf Notations:} In the following, we use $\mathbb{E}[\cdot]$ to denote the average over all random variables in $[\cdot]$, 
$\mathbb{E}_{X}[\cdot]$ to denote the expectation over the random variable $X$, and $\mathbb{P}(E)$ to denote the probability that an event $E$ occurs. Besides, $\|{\mathbf x}\|$ is used to represent the Euclidean norm between the coordinate ${\mathbf x}$ and the origin of the Euclidean space. $\bar{z}$ and $|z|$ denote the complex conjugate and modulus of the complex number $z$, respectively.  The notations $f_{X}(\cdot)$, $F_{X}(\cdot)$, $M_{X}(\cdot)$ and $\mathcal{L}_{X}(\cdot)$ are used to denote, respectively, the probability density function (PDF), cumulative distribution function (CDF), moment generating function (MGF), and Laplace transform of a random variable $X$. $\mathrm{erfc}(\cdot)$ is the complementary error function defined as $\mathrm{erfc}(x) = \frac{2}{\sqrt{\pi}} \int^{\infty}_{x} \exp (- t^2) \mathrm{d}t$. 

\section{Ambient Backscattering Assisted Wireless-Powered Communications} 

\subsection{Ambient Backscatter Communications} 
In backscatter communications, the information transmission is done by load modulation which does not involve active RF generation. In particular, a backscattering device tunes its antenna load reflection coefficient by switching between two or more impedances, resulting in a varying amount of incident signal to be backscattered. For example, when the impedance of the chosen load matches with that of the antenna, a little amount of the signal is reflected, exhibiting a signal absorbing state. Conversely, if the impedances are not matched, a large amount of the signal is reflected, indicating a signal reflecting state. A backscatter transmitter can use an absorbing state or reflecting state to transmit a `0' or `1' bit. Based on detection of the amount of the reflected signal, the transmitted information is interpreted at the receiver side.
 
Unlike conventional backscatter communication (e.g., for passive sensors~\cite{Ricardo2016Correia} and RFID tags~\cite{D.2010Dardari}), ambient backscattering functions without the need of a dedicated carrier emitter (e.g., RFID reader) to generate continuous waves. Instead, an ambient backscatter device utilizes exogenous RF waves as both energy resource to scavenge and signal resource to reflect. Moreover, ambient backscattering is featured with coupled backscattering and energy harvesting processes~\cite{X.2017Lu}. To initiate information transmission, the device first extracts energy from incident RF waves through rectifying. Once the rectified DC voltage is above the operating level of the circuit, the device is activated to conduct load modulation. In other words, modulated backscatter is generated on the reflected wave to transmit the encoded data, enabling a full-time transmission. For example, a recent experiment in~\cite{N.2014Parks} demonstrated that a 1 Mbps transmission could be achieved at the distance of 7 feet, when the incident RF power available is above -20 dBm.





\subsection{Hybrid D2D Communications}
\label{Hybrid_D2D_Communication} 
 
We now propose a novel hybrid D2D transmitter that combines two self-sustainable communication approaches, ambient backscatter communications and wireless-powered communications. On the one hand, ambient backscatter communications can be operated with very low power consumption. Thus, ambient backscattering may still be performed when the power density of ambient RF signals is low. On the other hand, the wireless-powered communications, also referred to as harvest-then-transmit (HTT)~\cite{H.2014Ju}, though have higher power consumption, can first accumulate harvested energy and then achieve longer transmission distance through active RF transmission. 
Therefore, these two approaches can well complement each other and result in better 
communication performance.

\begin{figure*}
\centering
\includegraphics[width=0.8\textwidth]{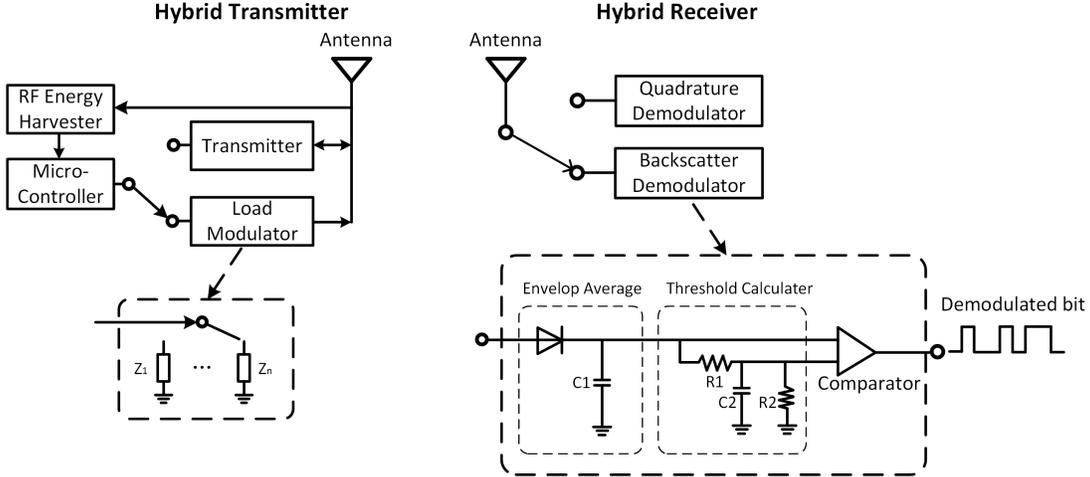}  
\caption{The structure of the hybrid transmitter and hybrid receiver.} 
\label{fig:hybrid_transmitter_receiver}
\end{figure*}

We depict the block diagrams of the hybrid transmitter and hybrid receiver in Fig.~\ref{fig:hybrid_transmitter_receiver}. The transmitter consists of the following main components: an antenna, used to scavenge energy and transmit information; an RF energy harvester to perform energy conversion from RF signals to direct current (DC); a load modulator, to perform load modulation utilizing incident radio waves from ambient; a low-power RF transmitter, for information transmission by active RF generation; 
and a low-power microcontroller, for the control of mode selection between RF transmitter and load modulator. 
With the designed architecture, the hybrid transmitter is flexible to perform active data transmission, backscattering, and RF energy harvesting.

At the receiver side, a dual-mode circuit as shown in Fig.~\ref{fig:hybrid_transmitter_receiver} can demodulate data from both the modulated backscatter and active RF transmission. The mode selection can be done by the transmitter through signaling. When the hybrid transmitter adopts HTT mode, a conventional quadrature demodulator composed of a phase shift module and a phase detector can be used. To demodulate from backscatter, the receiver adopts a simple circuit composed of three main components, namely, an envelop averager, a threshold calculator, and a comparator. The received signal is first smoothened to average out the variations due to embedded modulation. This step outputs  low and high voltage levels, which correspond to the time when the receiver observes only the ambient signal (i.e., an absorbing state), and the additional reflected signal (a reflecting state), respectively. Next, the threshold calculator computes a threshold by taking the mean of the two voltage levels. Consequently, by comparing the instantaneously generated voltage at the first step with the threshold, the comparator finally interprets the received signal into a stream of information bits.


As can be seen from Fig.~\ref{fig:hybrid_transmitter_receiver}, several common circuit
components, e.g., antenna, receiver, RF energy harvester and micro-controller, can be shared for the functions of active transmission and ambient backscattering.\footnote{Recently, a hybrid battery-powered transmitter with Bluetooth and backscattering functions, namely {\em Briadio}, is implemented in \cite{P.2016Hu}. The Briadio prototype demonstrates that integration of (active) Bluetooth and (passive) backscattering is practical and can be realized with low circuit complexity. Therefore, the integration of a wireless-powered active transmitter and an ambient backscatter transmitter can also be implemented with low complexity.} 
Thus, our proposed hybrid transmitter allows a tight integration of a wireless-powered transmitter and an ambient backscatter transmitter.
The hybrid transmitter also allows a highly flexible operation to perform either HTT mode or ambient backscattering mode.
Therefore, the hybrid transmitter needs to decide which mode to select when it wants to perform data transmission. 
We will investigate different mode selection protocols and analyze the corresponding performance in the sequel.

\section{Network Model and Stochastic Geometry Characterization} 
\subsection{Network Model} 
\label{sec:network_model}

We consider the hybrid D2D communications, introduced in Section \ref{Hybrid_D2D_Communication}, in coexistence with ambient RF transmitters, e.g., cellular base stations and mobiles. Fig.~\ref{fig:D2D_system} illustrates our considered system model. 
We consider two groups of coexisting ambient transmitters, denoted as $\Phi$ and $\Psi$, respectively, which work on different frequency bands. The RF energy harvester of the hybrid transmitter scavenges on the transmission frequency of $\Phi$. If the hybrid transmitter is in ambient backscattering mode, it performs load modulation on the incident signals from $\Phi$. Alternatively, when the hybrid transmitter is in HTT mode, it harvests energy from ambient transmitters in $\Phi$, and transmits over a different frequency band used by ambient transmitters in $\Psi$.\footnote{Similar to~\cite{X.April2014Lin}, we assume that the hybrid transmitter decides the transmit frequency and
indicates to the hybrid receiver through broadcasting in the preamble. Thus, the hybrid receiver is implemented to work on the transmit frequency of $\Phi$ and $\Psi$ when the hybrid transmitter is in ambient backscattering mode and HTT mode, respectively.}
The received signal at the hybrid receiver from the hybrid transmitter is impaired by the interference from $\Psi$.
We assume that $\Phi$ and $\Psi$ follow independent $\alpha$-Ginibre point process (GPP)~\cite{DecreusefondFlintVergne} which will be justified and detailed in Section~\ref{sec:geometricmodeling}. 
For example, the RF energy harvester of
the hybrid transmitter scavenges energy from LTE-A cellular mobiles on 1800 MHz. In HTT mode, the active D2D transmission is performed using WiFi Direct~\cite{Camps-Mur2013} over 2.4GHz, and gets interfered by the ambient users working on the same frequency band.  The locations of the 
ambient users on 1800MHz and those on 2.4GHz are independent. 

\begin{figure*}
	\centering
	\includegraphics[width=0.6\textwidth]{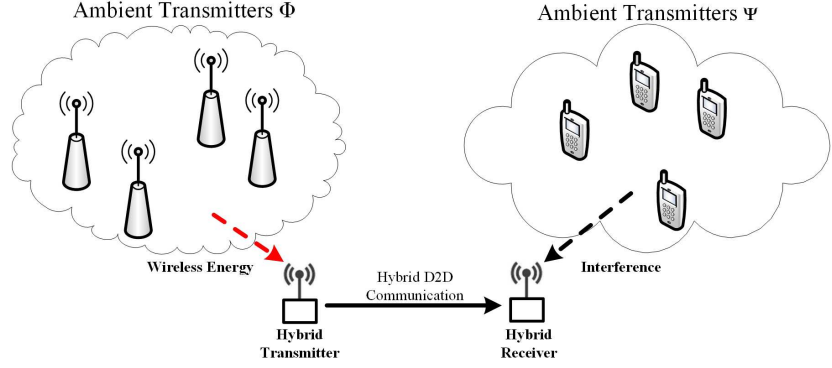}  
	\caption{Illustration of the hybrid D2D communication.} \label{fig:D2D_system}
\end{figure*} 


Without loss of generality, the hybrid transmitter, denoted as $\mathrm{S}$, and the associated hybrid receiver, denoted as $\mathrm{D}$, are assumed to locate at the origin when we analyze their corresponding performance, respectively. In particular, the point processes $\Phi$ and $\Psi$ are assumed to be supported on the circular observation windows $\mathbb{O}_{\mathrm{S}}$ and  $\mathbb{O}_{\mathrm{D}}$ with radius $R$, which are centered at  $\mathrm{S}$ and $\mathrm{D}$, respectively. 
The transmit power of the ambient transmitters belonging to $\Phi$ and $\Psi$ are denoted as $P_{A}$ and $P_{B}$, respectively. Let $\zeta_{A}$ and $\zeta_{B}$ denote the spatial density of $\Phi$ and $\Psi$, respectively. And $\alpha \in \big(0, 1\big]$ represents the repulsion factor which measures the correlation among the spatial points in $\Phi$ and $\Psi$.
Then, $\Phi$ can be represented by a homogeneous marked point process $\Phi=\{\mathbf{X}_{A},   \mathbf{C}_{A}, \mathcal{A}, \zeta_{A}, \alpha,  P_{A} \}$, where $\mathbf{X}_{A}=\{\mathbf{x}_{a} | a \in \Phi \}$ denotes the set of locations of the ambient transmitters in $\Phi$, $\mathbf{C}_{A} = \{c_{a} | a \in \Phi \}$ denotes the set of state indicators (in particular, $c_{a}=1$ if transmitter $a$ is on transmission in the reference time slot, and $c_{a} = 0$ otherwise), 
and $\mathcal{A}$ denotes the set of active ambient transmitters of $\Phi$ observed in $\mathbb{O}_{\mathrm{S}}$ by the hybrid transmitter.  
We assume that $c_{a}$ is an independent and identically distributed (i.i.d.) random variable. Then, the {\it transmission load} of $\Phi$ can be calculated as $l_{A}=\mathbb{P}[c_{a}=1]$, which measures the portion of time that an ambient transmitter is active. It is worth noting that the set of active transmitters in the reference time is a thinning point process with spatial density $l_{A} \zeta_{A}$. Similarly, $\Psi$ is characterized by $\Psi=\{\mathbf{X}_{B}, \mathbf{C}_{B}, \mathcal{B}, \zeta_{B}, \alpha, P_{B}\}$, where $\mathbf{X}_{B}$ denotes the set of the locations of transmitters in $\Psi$, $\mathbf{C}_{B}$ is the set of state indicators for $\Psi$,  and $\mathcal{B}$ denotes the set of the ambient transmitters of $\Psi$ observed in $\mathbb{O}_{\mathrm{D}}$ by the hybrid receiver.  $l_{B}=\mathbb{P}[c_{b}=1]$ denotes the transmission load of $\Psi$, where $c_{b}$ is the state indicator of $b \in \mathcal{B}$. Let $\xi$ represent the ratio of $l_{B}\zeta_{B}$ to $l_{A}\zeta_{A}$, i.e., $\xi=l_{B}\zeta_{B}/l_{A}\zeta_{A}$, referred to as the interference ratio. A larger value of $\xi$ indicates a higher level of interference. 

Let $\mathbf{x}_{\mathrm{S}}$ represent the location of the hybrid transmitter.
The power of the incident RF signals at the antenna of $\mathrm{S}$ can be calculated as
$P_{I}= P_{A} \sum_{a \in \mathcal{A}} h_{a,\mathrm{S}} \|\mathbf{x}_{a}-\mathbf{x}_{\mathrm{S}}\|^{-\mu}$, 
where $h_{x,y}$ represents the fading channel gain between $x$ and $y$ on the transmit frequency of $\Phi$, and $\mu$ denotes the path loss exponent.
The circuit of the hybrid transmitter becomes functional if it can extract sufficient energy from the incident RF signals. When the hybrid transmitter works in different modes (i.e., either HTT or ambient backscattering),  the hardware circuit consumes different amounts of energy.\footnote{The typical circuit power consumption rate of a wireless-powered transmitter ranges from hundreds of micro-Watts to several milli-Watts \cite{Y.2015Ishikawa
,X.June_2014Lu}, while that of a backscatter transmitter ranges from several micro-Watts to hundreds of micro-Watts~\cite{N.2014Parks}.} Let $\rho_{\mathrm{B}}$ and $\rho_{\mathrm{H}}$ denote the circuit power consumption rates (in Watt) in ambient backscattering and HTT modes, respectively. If the hybrid transmitter cannot harvest sufficient energy, an outage occurs.

In ambient backscattering mode, if the instantaneous energy harvesting rate (in Watt) exceeds $\rho_{\mathrm{B}}$, the hybrid transmitter can generate modulated backscatter. 
During backscattering process, a fraction of the incident signal power, denoted as $P_{H}$, is rectified for conversion from RF signal to direct current (DC), and the residual amount of signal power, denoted as $P_{R}$, is reflected to carry the modulated information. 
In ambient backscattering mode, the energy harvesting rate (in Watt) can be represented as~\cite{C.2012Boyer,C.2014Boyer}
$P^{\mathrm{B}}_{E}=\beta P_{H}=\beta \varrho P_{I}$,
where $0 < \beta \leq 1$ denotes the efficiency of RF-to-DC energy conversion, and $\varrho$ represents the fraction of the incident RF power for RF-to-DC energy conversion. Note that the value of $\varrho$ depends on the symbol constellation adopted for multi-level load modulation~\cite{C.2012Boyer}. For example, $\varrho$ is 0.625 on average assuming equiprobable symbols if binary constellations are adopted with modulator impedance values set as 0.5 and 0.75~\cite{C.2014Boyer}.

Let $\mathbf{x}_{\mathrm{D}}$ represent the location of the hybrid receiver.
$d\!\!=\!\!\|\mathbf{x}_{\mathrm{S}}\!-\! \mathbf{x}_{\mathrm{D}}\|$ denotes the distance between $\mathrm{S}$ and $\mathrm{D}$. 
Then, in ambient backscattering mode, the power of the received backscatter at $\mathrm{D}$ from $\mathrm{S}$
can be calculated as $P_{\mathrm{S},\mathrm{D}} \!=\!
    \delta P_{I}  (1-\varrho) h_{\mathrm{S},\mathrm{D}} d^{-\mu}$ if $P^{\mathrm{B}}_{E} > \rho_{\mathrm{B}}$ and $P_{\mathrm{S},\mathrm{D}} = 0$ otherwise.  
Here $0<\delta \leq 1$ is the backscattering efficiency of the transmit antenna, which is related to the antenna aperture~\cite{V.2006Nikitin}.  
If $\mathrm{S}$ is active in ambient backscattering mode, the resulted SNR 
at $\mathrm{D}$ is
\begin{equation} \label{eqn:SNR_B}
	\nu_{\mathrm{B}} =\frac{P_{\mathrm{S},\mathrm{D}}  }{ \sigma^2}=\frac{\delta P_{I} (1-\varrho) h_{\mathrm{S},\mathrm{D}}}{ d^{\mu}\sigma^2},
\end{equation}
where $\sigma^2$ is the variance of additive white Gaussian noise (AWGN).

If the received SNR $\nu_{\mathrm{B}}$ is above a threshold $\tau_{\mathrm{B}}$, $\mathrm{D}$ is able to successfully decode information from the modulated backscatter at a pre-designed rate $T_{\mathrm{B}}$ (in bits per second (bps)). This backscatter transmission rate is dependent on the setting of resistor-capacitor circuit elements. For example, it has been demonstrated in~\cite{V.2013Liu} that 1~kbps and 10 kbps backscatter transmission rates can be achieved if the values of circuit elements, i.e., R1, R2, C1, and C2, in Fig.~\ref{fig:hybrid_transmitter_receiver} are set as (150 kOhm, 10 MOhm, 4.7 nano-farad, 10 nano-farad) and (150 kOhm, 10 MOhm, 680 nano-farad, 1 micro-farad), respectively.


When the hybrid transmitter $\mathrm{S}$ chooses to adopt active RF transmission, it is operated by the HTT protocol~\cite{H.2014Ju}. In HTT mode, the hybrid transmitter works in a time-slot based manner. Specifically, in each time slot, the first period, with time fraction $\omega$, is for harvesting energy, during which the impedance of the load modulator is tuned to fully match that of the antenna to maximize the energy conversion efficiency. The corresponding energy harvesting rate is $P^{\mathrm{H}}_{E}=\omega \beta P_{I}$.
This harvested energy is first utilized to power the circuit. Then the remaining energy, if available, is stored in an energy storage. 
If the harvested energy is enough to operate the circuit, the hybrid transmitter spends the rest of the period $(1-\omega)$ to perform active transmission with the stored energy.  
     
In the active transmission phase, the transmit power of $\mathrm{S}$ is $P_{\mathrm{S}} =  \frac{ P^{\mathrm{H}}_{E}- \rho_{\mathrm{H}}}{1-\omega }$ if $P^{\mathrm{H}}_{E} > \rho_{\mathrm{H}}$ and $P_{\mathrm{S}} = 0$ otherwise.
Then, the received signal-to-interference-plus-noise ratio (SINR) at $\mathrm{D}$ can be expressed as
\begin{eqnarray}
\nu_{\mathrm{H}} =\frac{ P_{\mathrm{S}} \widetilde{h}_{\mathrm{S},\mathrm{D}} d^{-\mu} }{\sum_{b \in \mathcal{B}}  P_{B} \widetilde{h}_{b,\mathrm{D}}\|\mathbf{x}_{b}-\mathbf{x}_{\mathrm{D}}\|^{-\mu} +\sigma^{2} },
\end{eqnarray}
where $\widetilde{h}_{x,y}$ denotes the fading channel gain between $x$ and $y$ on the transmit frequency of $\Psi$. 

As the hybrid D2D communications and the transmission from ambient transmitters may occur in different environments,
we consider different fading channels for $h_{\mathrm{S},\mathrm{D}}$, $\widetilde{h}_{\mathrm{S},\mathrm{D}}$, $h_{a,\mathrm{S}}$ and $\widetilde{h}_{b,\mathrm{D}}$. Specifically,  $h_{\mathrm{S},\mathrm{D}}$ and $\widetilde{h}_{\mathrm{S},\mathrm{D}}$ are assumed to follow Rayleigh fading. Both $h_{a,\mathrm{S}}$ and $\widetilde{h}_{b,\mathrm{D}}$ follow i.i.d. Nakagami-$m$ fading, which is a general channel fading model that contains Rayleigh distribution as a special case when $m=1$. 
This channel model allows a flexible evaluation of the impact of the ambient signals.\footnote{Our work can be extended to the case when $h_{\mathrm{S},\mathrm{D}}$ and $\widetilde{h}_{\mathrm{S},\mathrm{D}}$ also follow a Nakagami-$m$ distribution. However, the resulted analytical expressions bring about high computational complexity without much insight. Therefore, we focus on exponentially distributed $h_{\mathrm{S},\mathrm{D}}$  and $\widetilde{h}_{\mathrm{S},\mathrm{D}}$ in this paper.}  Let $\mathcal{G}(x,y)$ represent the gamma distribution with shape parameter $x$ and scale parameter $y$, and $\mathcal{E}(x)$ represent the exponential distribution with rate parameter $x$. Thus, the fading channel gains are expressed as  $h_{a,\mathrm{S}}, \widetilde{h}_{b,\mathrm{D}} \sim \mathcal{G}( m, \theta/m)$ and $h_{\mathrm{S},\mathrm{D}}, \widetilde{h}_{\mathrm{S},\mathrm{D}} \sim \mathcal{E}(\lambda)$, where $\theta$ and $\lambda$ are  expectation of the corresponding fading channel gains. 

Let $W$ denote the frequency bandwidth for active transmission in HTT mode. The transmission capacity of a hybrid  transmitter in HTT mode can be computed as $\mathcal{T}_{\mathrm{H}}  \!=\!  (1-\omega)W \log_{2} \left(1\!+\! \nu_{\mathrm{H}} \right)$ if  $P^{\mathrm{H}}_{E} > \rho_{\mathrm{H}}$ and $\nu_{\mathrm{H}} > \tau_{\mathrm{H}}$, and $\mathcal{T}_{\mathrm{H}} = 0$ otherwise. 
Here $\tau_{\mathrm{H}}$ is the minimum SINR threshold for the hybrid receiver to successfully decode from the received active RF signals~\cite{H.2013ElSawy}. 

For operation of our proposed hybrid transmitter, we consider two mode selection protocols, namely, \emph{power threshold-based protocol} (PTP) and \emph{SNR threshold-based protocol} (STP). 
\begin{itemize}

\item Under PTP, a hybrid transmitter first detects the available energy harvesting rate $P^{\mathrm{H}}_{E}$. If $P^{\mathrm{H}}_{E}$ is below the threshold which is needed to power the RF transmitter circuit (for active transmission), i.e., $ P^{\mathrm{H}}_{E} \leq \rho_{\mathrm{H}}  $, ambient backscattering mode will be used. Otherwise, HTT mode will be adopted. 

\item Under STP, the hybrid transmitter first attempts to transmit by backscattering. If the achieved SNR at the receiver is above the threshold which is needed to decode information from the backscatter, i.e., $\nu_{\mathrm{B}} > \tau_{\mathrm{B}}$, the transmitter will be in ambient backscattering mode. Otherwise, it will switch to HTT mode.

\end{itemize}


The motivation behind PTP is to use active transmission for higher throughput if the ambient energy resource is abundant, and adopt backscattering to diminish the occurrences of energy outage otherwise. The motivation of STP is to enjoy full-time transmission by backscattering when the achievable SNR is high, and adopt HTT if ambient backscattering does not have good performance. 
Note that for implementation of the two protocols, PTP allows the transmitter to operate independently based on its local information while STP requires the transmitter to obtain feedback from the receiver.

We investigate the two simple protocols described above in view of their practicality of implementation and tractable analysis. We will reveal how the naive mechanism adopted in each protocol affects different performance metrics. More sophisticated protocols that offer superior performance can be designed by utilizing system information such as channel state information feedback, interference detection, and energy source localization. However, these protocols may require more computational overhead as well as complicated and expensive hardware implementation, which are not practical for low-power devices based on energy harvesting. 
 
{\bf Remark 1:} The analytical expressions derived in this paper represent a lower bound on the achievable performance. This is because for implementation simplicity and practicability we consider that the mode selection is performed only once at the beginning of hybrid D2D communication. The selected mode may not always be the better choice when the network channel condition varies. 
The analytical approach presented in the paper can be straightforwardly extended to the case when mode selection is performed at the beginning of each fading block.

\subsection{Geometric Modeling of the Systems} \label{sec:geometricmodeling}
 

Due to its tractability, the Poisson point process (PPP) has been widely adopted for modeling different types of wireless networks~\cite{H.2013ElSawy}. PPP abstracts each randomly located point according to a uniform distribution in the Euclidean space.  
However, as pointed out in \cite{G.2011Andrews}, PPP modeling only serves as lower bounds to the coverage probability and mean rate of real-world deployment. 
The reason is that the spatial points in a PPP may locate very close to each other because of independence.
This calls for the need of more sophisticated and general geometric approaches to model the correlation among spatial points.
In this context, GPP and its variants have attracted considerable attention. 
Recent research work has adopted GPP in~\cite{N.2015Miyoshi}, $\alpha$-GPP in~\cite{I.2015Flint ,X.Sept_2017Lu,X.March2015Lu,I.December_2014Flint} and $\beta$-GPP in~\cite{N.2015Deng 
} to model the distribution of cellular base stations.
In this paper, the performance analysis of the hybrid D2D communications is based on $\alpha$-GPP~\cite{DecreusefondFlintVergne}.
$\alpha$-GPP is a repulsive point process which allows to characterize the repulsion among randomly located points and covers the PPP as a special case (i.e., when $\alpha \to 0$).  The coefficient $\alpha$ ($\alpha = -1/\kappa$ for a positive integer $\kappa$) indicates the repulsion degree of the spatial points.  Specifically, the repulsion is the strongest in case $\alpha=-1$ and disappears when $\alpha$ approaches 0. 
In this paper, we use $\alpha$-GPP because it renders tractable analytical expressions in terms of Fredholm determinants. The Fredholm determinant is a generalized determinant of a matrix defined by bounded operators on a Hilbert space and has shown to be efficient for numerical evaluation of the relevant quantities~\cite{L.2015Decreusefond}. 

In the following, we describe some fundamental features and properties of $\alpha$-GPP which will be applied later in the analysis of this paper.
For any $\alpha$-GPP $\Omega$, 
let $\zeta$ denote the spatial density of the points of $\Omega$, and $\mathcal{K}$ represent an almost surely finite collection of $\Omega$ located inside an observation window $\mathbb{O}_{\mathbf{x}}$, denoted as a circular Euclidean plane with positive radius $R$.  Without loss of generality, in this paper, we restrict the analysis on a generic point located at $\mathbf{x}$ within  $\mathbb{O}_{\mathbf{x}}$. We begin with the Laplace transform of $\alpha$-GPP  characterized by means of Fredholm determinants~\cite{ShiraiTakahashi}.  
The Fredholm determinant is generally expressed in the form of a complex-valued function, which contains the coordinates of the spatial points represented by complex numbers as the variables. 
For $|\alpha| \leq 1$, the Fredholm determinant of an arbitrary function $F$ is expressed as $\mathrm{Det}\big(\mathrm{Id} + \alpha F\big)$. The readers are referred to \cite{ShiraiTakahashi} for the mathematics and properties of  the Fredholm determinant. 

\begin{prop}  {\em \cite[Theorem 2.3]{L.2015Decreusefond} } 
\label{lemma1} Let $\varphi$ 
represent an arbitrary real-valued function. For an $\alpha$-GPP, the Laplace transform of $\sum_{k \in \mathcal{K}} \varphi(\mathbf{x}_{k})$ can be expressed as
\begin{eqnarray}
	\label{eq:laplacedpp} 
	\mathbb E \left[ \exp\left(-s \sum_{k\in\mathcal K}\varphi(\mathbf x_k)\right)\right] = \mathrm{Det}\big(\mathrm{Id}+\alpha \mathbb{K}_\varphi(s)\big)^{-\frac{1}{\alpha}},
\end{eqnarray}
	where $\mathbb{K}_\varphi(s)$ is given by
	\begin{align}
 & \mathbb{K}_\varphi(s) =	\sqrt{1-\exp (-s\varphi(\mathbf{x}))} G_{\Omega} (\mathbf x,\mathbf y) \sqrt{1-\exp ({-s\varphi(\mathbf y)})},   \quad \mathbf x,\mathbf y\in \mathcal{K}, \label{eqn:kernel_K} 
	\end{align}
	wherein $G_{\Omega}$ is the Ginibre kernel which represents the correlation force among different spatial points in $\Omega$ defined as 
	\begin{equation}
	\label{eq:ginibre}
	G_{\Omega} (\mathbf x,\mathbf y)=\zeta\,e^{\pi\zeta \mathbf x \bar{\mathbf y}} e^{-\frac{\pi\zeta}{2}( |\mathbf x|^2 + |\mathbf y|^2)},
	\quad 
	\mathbf x,\mathbf y \in \mathcal{K}.
	\end{equation}
\end{prop}



As the Laplace transform in (\ref{eq:laplacedpp}) is given in the form of Fredholm determinant, the evaluation of it may involve high computation complexity. For example, the conventional approach in~\cite{F.2008Bornemann} approximates 
the Fredholm determinant by the determinant of an $N \times N$ matrix, resulting in a complexity of $O(N^3)$. 
The recent results  in~\cite{B.2017Kong} allow a more efficient computation of the Fredholm determinant 
with significantly reduced complexity. A simplified expression for evaluating the Fredholm determinant 
is presented in the following Proposition.

\begin{prop} {\em \cite[Lemma 3]{B.2017Kong}} 
With $\mathbb{K}_{\varphi}(s)$  defined in (\ref{eqn:kernel_K})  and  
$G_{\Omega}(\mathbf{x},\mathbf{y})$  defined in~(\ref{eq:ginibre}), the Fredholm determinant on the right hand side of (\ref{eq:laplacedpp}) can be evaluated as
\begin{align} \label{evaluation_det}
& \mathrm{Det} \big(\mathrm{Id}+ \alpha \mathbb{K}_{\varphi}(s)\big)^{-\frac{1}{\alpha}} \!\!=\!\!\!\!\prod^{N_{\mathrm{closed}}}_{n=0}\!\!  \bigg( \! 1 \!+\! \frac{2\alpha (\pi \zeta)^{n+1}}{n!} \! \int^{R}_{0} \!\!  \exp(-\pi \zeta r^2) 
r^{2n+1} \big( 1\!-\! \exp \! \big(\! \!-\!s \varphi(r)\big)   \big) \mathrm{d} r \bigg)^{ \!\!\!-\frac{1}{\alpha}} . \hspace{-2mm}
\end{align}

\end{prop}
The complexity in calculating (\ref{evaluation_det}) is $O(N_{\mathrm{closed}})$. As $N_{\mathrm{closed}}$ goes to infinity, the exponential convergence rate of (\ref{evaluation_det}) follows from the smoothness of the Ginibre kernel~\cite{DecreusefondFlintVergne}.


\subsection{Performance Metrics}
\label{sec:Performance_Metrics}
We measure the performance of the hybrid D2D communications in three important metrics, namely, energy outage probability, coverage probability, and throughput.

The hybrid transmitter experiences an energy outage when the energy obtained from the ambient transmitters is not enough to support its circuit operation. 
Let $\mathcal{O}_{\mathrm{B}}$ and $\mathcal{O}_{\mathrm{H}}$ denote the energy outage probability of the hybrid transmitter being in ambient backscattering mode and HTT mode, respectively. Mathematically, the overall energy outage probability is given as
\begin{align} \label{def:energyoutage}
\mathcal{O} &  =   \mathcal{B} \hspace{0.5mm} \mathcal{O}_{\mathrm{B}}   +   (1 - \mathcal{B} )\mathcal{O}_{\mathrm{H}} 
= \mathcal{B}  \hspace{0.5mm} \mathbb{P}[P^{\mathrm{B}}_{E} \leq \rho_{\mathrm{B}}] + (1-\mathcal{B} )  \mathbb{P}[P^{\mathrm{H}}_{E} \leq \rho_{\mathrm{H}}],  
\end{align}
where $\mathcal{B}$ denotes the probability that the hybrid transmitter selects ambient backscattering mode. 

The transmission of the hybrid transmitter is considered to be successful if the achieved SNR or SINR at the associated receiver exceeds its target threshold. We define coverage as an event of successful transmission.  
Let $\mathcal{C}_{\mathrm{B}}$ and $\mathcal{C}_{\mathrm{H}}$ denote the coverage probability of the hybrid transmitter being in ambient backscattering mode and HTT mode, respectively. Then, the overall coverage probability is given as 
\begin{align} \label{def:coverage_probability} 
\mathcal{C}  &  = \mathcal{B} \hspace{0.5mm} \mathcal{C}_{\mathrm{B}} + (1-\mathcal{B} )\mathcal{C}_{\mathrm{H}} 
= \mathcal{B} \hspace{0.5mm} \mathbb{P}[\nu_{\mathrm{B}}> \tau_{\mathrm{B}}, P^{\mathrm{B}}_{E} > \rho_{\mathrm{B}}] 
+ (1-\mathcal{B} )  \mathbb{P}[\nu_{\mathrm{H}}> \tau_{\mathrm{H}}, P^{\mathrm{H}}_{E} > \rho_{\mathrm{H}}]. 
\end{align} 

Moreover, the average throughput achieved by the hybrid transmitter is given as
\begin{eqnarray}\label{def:throughput}
	\mathcal{T} = \mathcal{B} \hspace{0.5mm} \mathcal{T}_{\mathrm{B}} 
	+ (1-\mathcal{B} )  \mathcal{T}_{\mathrm{H}}, 
\end{eqnarray}
where $\mathcal{T}_{\mathrm{B}}$ denotes the average throughput in ambient backscattering mode and $\mathcal{T}_{\mathrm{H}}$ has been defined in Subsection~\ref{sec:network_model}.


An upper bound on the achievable performance can be obtained by considering block fading channels with mode selection performed at the beginning of each fading block. As we focus on the impact of system parameters and comparison of the proposed mode selection protocols, we omit presenting the upper bound. The upper bound performance can be derived by following the same analytical approach presented in the paper.

 
\section{Analytical Results}
\label{sec:Analytical_Results}
 
In this section, we derive analytical expressions for the performance metrics introduced in Section~\ref{sec:Performance_Metrics} based on the repulsive point process framework introduced in Section~\ref{sec:geometricmodeling}. 
 
\subsection{Energy Outage Probability} 
We first derive the expressions of the energy outage probability based on the definition in (\ref{def:energyoutage}). 
 
\begin{theorem}
\label{thm:Outage_PTP}
Under PTP, the energy outage probability of a hybrid transmitter is calculated as
\begin{align}
\label{eq:energy_outage_probability_PTP}
\mathcal{O}_{\mathrm{PTP}} & = 
F_{P_{I}} \Big( \frac{\rho_{\mathrm{H}}}{\omega \beta }\Big) \left( F_{P_{I}} \Big( \frac{\rho_{\mathrm{B}}}{ \beta \varrho }\Big) - F_{P_{I}} \Big( \frac{\rho_{\mathrm{H}}}{\omega \beta }\Big) + 1 \right),
\end{align} 
where $F_{P_{I}}(\rho)$ is the CDF of $P_{I}$ given as 
\begin{equation} \label{CDF_PI}
F_{P_{I}}(\rho)=\mathcal{L}^{-1}\left\{ \frac{\mathrm{Det}\big(\mathrm{Id}+\alpha \mathbb{A}_{\Phi}(s )\big)^{-\frac{1}{\alpha} }}{s} \right\}(\rho),
\end{equation}
wherein $\mathcal{L}^{-1}$ means inverse Laplace transform and $\mathbb{A}_{\Phi}(s)$ is 
given by
\begin{align}
\label{eqn:kernal_A}
& \mathbb{A}_{\Phi}(s)=\sqrt{1-\left( 1+ \frac{s\theta P_{A}}{m\|\mathbf{x}-\mathbf{x}_{\mathrm{S}}\|^{\mu}} \right)^{\!-m}} 
G_{\Phi}(\mathbf{x},\mathbf{y}) \sqrt{1-\left( 1+ \frac{s\theta P_{A}}{m\|\mathbf{y}-\mathbf{x}_{\mathrm{S}}\|^{\mu}} \right)^{\!-m}}, 
\end{align} 
and $G_{\Phi}$ is the Ginibre kernel of $\Phi$ defined as
\begin{align}
\label{eq:ginibre_Phi}
\hspace{-2mm} G_{\Phi}(\mathbf x,\mathbf y)  = l_{A} \zeta_{A} \,e^{\pi l_{A} \zeta_{A} \mathbf x \bar{\mathbf y}} e^{-\frac{\pi l_{A} \zeta_{A} }{2}( |\mathbf x|^2 + |\mathbf y|^2)}, 
 \mathbf x,\mathbf y \in \mathcal{A}. \hspace{-2mm}
\end{align} 
\end{theorem} 
 
For readability, we present the proof of {\bf Theorem} \ref{thm:Outage_PTP} in Appendix I.

Consequently, we extend the above outcome in {\bf Theorem}~\ref{thm:Outage_PTP} to the case of STP by altering the mode selection probability based on the STP criteria, resulting in the following Theorem.

\begin{theorem}
\label{thm:PowerOutage_STP}
Under STP, the energy outage probability of a hybrid transmitter is
\begin{align}
\label{eq:energy_outage_probability_STP}
\mathcal{O}_{\mathrm{STP}}= & \int^{\infty}_{\frac{\rho_{\mathrm{B}}}{\beta \varrho}} \exp   \left( - \frac{\lambda \tau_{\mathrm{B}}d^{\mu} \sigma^2 }{\delta \rho \left ( 1- \varrho \right ) } \right) f_{P_{I}}(\rho) \mathrm{d} \rho 
\left( F_{P_{I}} \Big( \frac{\rho_{\mathrm{B}}}{\beta \varrho}\Big)  - F_{P_{I}} \Big( \frac{\rho_{\mathrm{H}}}{\omega \beta }\Big) \right) + F_{P_{I}} \Big( \frac{\rho_{\mathrm{H}}}{\omega \beta }\Big),
\end{align} 
where $F_{P_{I}}(\rho)$ has been given in (\ref{CDF_PI}), and $f_{P_{I}}(\rho)$ is the PDF of $P_{I}$ calculated as 
\begin{equation}
f_{P_{I}}(\rho) =\mathcal{L}^{-1} \left\{ \mathrm{Det}\big(\mathrm{Id}+\alpha \mathbb{A}_{\Phi}(s)\big)^{-\frac{1}{\alpha} }\right\}(\rho), \label{eq:PDF}
\end{equation} 
wherein $\mathbb{A}_{\Phi}(s)$ has been defined in (\ref{eqn:kernal_A}).
\end{theorem} 
 
The proof of Theorem~\ref{thm:PowerOutage_STP} is shown in Appendix II.

 
Note that both $\mathcal{O}_{\mathrm{PTP}}$ and $\mathcal{O}_{\mathrm{STP}}$ are functions of $\zeta_A$, not $\zeta_B$. Thus, given $\zeta_A$ and the transmission load $l_A$, the interference ratio $\xi$ does not affect the energy outage probability. 
We also note that similar to the stochastic geometry analysis based on PPP in the existing literature, e.g., \cite{G.2011Andrews}, it is difficult to see the relationship between the performance metric and system parameters directly from the general-case results in Theorems~\ref{thm:Outage_PTP} and \ref{thm:PowerOutage_STP} derived based on the $\alpha$-GPP framework. However, these general-case results can be simplified in some special cases. 
We then investigate a special setting which considerably simplifies the above results. 
 
\begin{corollary} \label{corollary2} 
When the distribution of ambient transmitters in $\Phi$ follows a PPP, the RF signals from these transmitters experience Rayleigh fading (i.e., $h_{a,\mathrm{S}} \sim \mathcal{E}(1)$), and the path loss exponent is equal to 4, 
the energy outage probability of a hybrid transmitter can be evaluated by (\ref{eq:energy_outage_probability_PTP}) under PTP and (\ref{eq:energy_outage_probability_STP}) under STP, with $f_{P_{I}}(\rho)$ and $F_{P_{I}}(\rho)$ expressed, respectively, as
\begin{align}
f_{P_{I}}(\rho) = \frac{ 1 }{4}\! \left(\frac{\pi}{\rho}\right)^{\!\!\frac{3}{2}} \!\zeta_{A}\! \sqrt{\!P_{A}} \exp \left(\!- \frac{\pi^4 \zeta_{A}^2 P_{A}}{16 \rho} \!\right),
\end{align}
and
\begin{align}\label{corollary:PDF_no_replusion_Rayleigh}
F_{P_{I}}(\rho) = \mathrm{erfc} \left(\!\frac{ \zeta_{A}\sqrt{P_{A}} \pi^2}{4 \sqrt{\rho}} \!\right). 
\end{align}

\end{corollary}

The proof of {\bf Corollary} \ref{corollary2} is given in Appendix III. 
 
\subsection{Coverage Probability} 

Next, we consider the coverage probability between a hybrid D2D transmitter-receiver pair. 
We have the coverage probability of PTP described as follows.
\begin{theorem}
\label{thm:CoverageOutage_PTP}
The coverage probability of the hybrid D2D communications under PTP is
\begin{align}
\label{eq:coverage_probability_Rayleigh_PTP}
\mathcal{C}_{\mathrm{PTP}} &  = \left(1-  F_{P_I}\Big(\frac{\rho_{\mathrm{H}}}{\omega \beta}\Big) 
\right) \! \int^{\infty}_{\frac{\rho_{\mathrm{H}}}{\beta \omega}}  \exp \left( - \frac{\lambda \tau_{\mathrm{H}}d^{\mu} \sigma^{2} (1-\omega)}{\omega \beta \rho - \rho_{\mathrm{H}}} \! \right) 
\mathrm{Det}\big( \mathrm{Id} + \alpha \mathbb{B}_{\Psi}(\rho) \big)^{-\frac{1}{\alpha}} f_{P_{I}}(\rho) \mathrm{d}\rho   + F_{P_I}\Big(\frac{\rho_{\mathrm{H}}}{\omega \beta}\Big) \nonumber \\ &  \hspace{90mm} \times
\int^{\infty}_{\frac{\rho_{\mathrm{B}}}{\beta \varrho}} \exp \left( - \frac{\lambda \tau_{\mathrm{B}}d^{\mu} \sigma^2 }{\delta \rho \left (1- \varrho \right ) } \right) f_{P_{I}}(\rho) \mathrm{d}\rho, 
\end{align} 
where $F_{P_{I}}(\rho)$ and $f_{P_{I}}(\rho)$ have been obtained in (\ref{CDF_PI}) and (\ref{eq:PDF}), respectively, 
and $\mathbb{B}_{\Psi}(\rho)$ is
\begin{align} \label{eqn:kernel_B}
& 
\mathbb{B}_{\Psi} (\rho)   =   \!  \sqrt{   1 \!  - \!  \left(\!   1\!  +\!  \frac{\theta \lambda \tau_{\mathrm{H}}d^{\mu} (1\!-\!\omega)  P_{B} }{m(\omega \beta \rho\! - \! \rho_{\mathrm{H}})\|\mathbf{x}\!-\!\mathbf{x}_{\mathrm{D} }\|^{\mu} } \!  \right)^{\!\!\! -m} } 
G_{\Psi} (\mathbf{x},\mathbf{y}) \sqrt{ 1\!-\! \left( \! 1\! +\! \frac{\theta \lambda \tau_{\mathrm{H}}d^{\mu} (1\!-\!\omega)   P_{B} }{m(\omega \beta \rho\! - \!\rho_{\mathrm{H}})\|\mathbf{y}\!-\!\mathbf{x}_{\mathrm{D}}\|^{\mu} }\!\! \right)^{\!\!\!-m} },
\end{align}
wherein $G_{\Psi}$ is the Ginibre kernel of $\Psi$ defined as
\begin{align}
\hspace{-1.5mm} G_{\Psi}(\mathbf x,\mathbf y) \! = \! l_{B} \zeta_{B} \,e^{\pi l_{B} \zeta_{B} \mathbf x \bar{\mathbf y}} e^{-\frac{\pi l_{B} \zeta_{B}}{2}( |\mathbf x|^2 + |\mathbf y|^2)}, 
 \mathbf x,\mathbf y \in \mathcal{B}.	\hspace{-1mm} 
 \label{eq:ginibre_Psi}
\end{align}

\end{theorem} 
 
The proof of {\bf Theorem} \ref{thm:CoverageOutage_PTP} is shown in Appendix IV.

Moreover, we derive the coverage probability for STP in the following Theorem. 
\begin{theorem}
\label{thm:CoverageOutage_STP}
The coverage probability of the hybrid D2D communications under STP is
\begin{align}
\label{eq:coverage_probability_STP} & \hspace{-2mm} \mathcal{C}_{\mathrm{STP}}  = \! \int^{\infty}_{\frac{\rho_{\mathrm{H}}}{\beta \omega} } \exp \left(\! - \frac{\lambda \tau_{\mathrm{H}}d^{\mu} (1\!-\!\omega)\sigma^{2} }{\omega \beta \rho - \rho_{\mathrm{H}} }  \right) \mathrm{Det}\big( \mathrm{Id} + \alpha \mathbb{B}_{\Psi} (\rho) \big)^{\!-\frac{1}{\alpha}}  
f_{P_{I}}(\rho) \mathrm{d} \rho \times \int^{\frac{\rho_{\mathrm{B}}}{\beta\varrho}}_{0} \exp \left( - \frac{\lambda \tau_{\mathrm{B
}}d^{\mu} \sigma^2 }{\delta \rho(1 - \varrho)} \right) f_{P_{I}}(\rho) \mathrm{d}\rho   \nonumber \\
&  \hspace{95mm}  
+ \left [ \int^{\infty}_{\frac{\rho_{\mathrm{B}}}{\beta \varrho}} \exp \left( - \frac{\lambda \tau_{\mathrm{B}}d^{\mu} \sigma^2 }{\delta \rho(1 - \varrho)} \right) f_{P_{I}}(\rho) \mathrm{d}\rho \right]^2, 
\end{align}
where 
$f_{P_{I}}(\rho)$ has been obtained in (\ref{eq:PDF}), and $\mathbb{B}_{\Psi} (\rho)$ is defined in (\ref{eqn:kernel_B}).

\end{theorem} 

 
\begin{proof}
According to the criteria of STP, $\mathcal{C}_{\mathrm{STP}}$ can be expressed by $\mathcal{C}_{\mathrm{PTP}}$ in (\ref{eqn:overall_coverage_probability}) with $\mathcal{B}_{\mathrm{PTP}}$ replaced by $\mathcal{B}_{\mathrm{STP}}$ given in (\ref{eqn:delta_B_STP}). Therefore, (\ref{eq:coverage_probability_STP}) can be obtained from (\ref{eq:coverage_probability_Rayleigh_PTP}) through the aforementioned replacement. 
\end{proof} 

\subsection{Throughput} 

Then, we move on to calculate the average throughput that can be achieved over a hybrid D2D communication link. 
We have the average throughput of PTP presented as follows:

\begin{theorem}
\label{thm:throughput}
Under PTP, the average throughput of a hybrid D2D communication link can be computed as
\begin{align}
& \hspace{-2mm} \mathcal{T}_{\mathrm{PTP}}  = T_{\mathrm{B}}  F_{ P_{I}} \Big( \frac{\rho_{\mathrm{B}}}{\omega \beta}\Big) \!\!
 \int^{\infty}_{\frac{\rho_{\mathrm{B}}}{\beta \varrho}} \! \exp \left( - \frac{\lambda \tau_{\mathrm{B}}d^{\mu} \sigma^2 }{\delta \rho(1 - \varrho)} \right) f_{P_{I}}(\rho) \mathrm{d}\rho 
 + \! (1\!-\!\omega) W \!\! \left( \! 1\!-\!F_{P_{I}}\Big(\frac{\rho_{\mathrm{B}}}{\omega \beta}\Big) \! \! \right)  \nonumber \\ & \hspace{20mm}\times \int^{\infty}_{ \log_{2} (1+\tau_{\mathrm{H}} \!)}  \int^{\infty}_{\frac{\rho_{\mathrm{H}}}{\beta\omega}}  \! \!\mathrm{Det} \big(\mathrm{Id}\!+\!\alpha \mathbb{C}_{\Psi}(\rho)\!\big)^{\!-\frac{1}{\alpha} }    \exp \left( \!- \frac{ \lambda d^{ \mu}\sigma^{2} (1-\omega) (2^t-1) }{ \omega \beta \rho - \rho_{\mathrm{H}} }  \!\right) f_{P_{I}}(\rho) \mathrm{d}\rho \mathrm{d}t,  
  \label{eq:throughput_PTP}
\end{align} 
where 
$F_{P_{I}}(\rho)$ and $f_{P_{I}}(\rho)$ have been obtained in (\ref{CDF_PI}) and (\ref{eq:PDF}), respectively,
and $\mathbb{C}_{\Psi} (\rho)$ is computed as 
\begin{align} \label{eqn:kernel_C}
&  \mathbb{C}_{\Psi} (\rho) = \! \sqrt{ 1 \!-\! \left( \!1 \!+\! \frac{\theta \lambda d^{\mu} (2^t\!-\!1) (1\!-\!\omega)   P_{B} }{m(\omega \beta \rho \!-\! \rho_{\mathrm{H}})\|\mathbf{x}\!-\!\mathbf{x}_{\mathrm{D}}\|^{\mu} \!}  \right)^{\!\!\!-m} }  
G_{\Psi} (\mathbf{x},\!\mathbf{y})  \sqrt{ 1 \!-\!  \left(\!1 \!+\! \frac{\theta \lambda d^{\mu} (2^t\!-\!1) (1\!-\!\omega)  P_{B} }{m(\omega \beta \rho \!-\! \rho_{\mathrm{H}})\|\mathbf{y}\!-\!\mathbf{x}_{\mathrm{D}}\|^{\mu} } \right)^{\!\!\!-m} }. \hspace{-2mm}
\end{align}

\end{theorem} 

The proof of {\bf Theorem} \ref{thm:throughput} is presented in Appendix V. 

Consequently, utilizing $\mathcal{B}_{\mathrm{STP}}$ obtained in (\ref{eqn:delta_B_STP}), 
we arrive at the following theorem stating the achievable throughput for STP.

\begin{theorem}
\label{thm:throughput_STP}
Under STP, the average throughput of a hybrid D2D communication link can be computed as
\begin{align} 
\mathcal{T}_{\mathrm{STP}} &
= T_{\mathrm{B}}  \left[  \int^{\infty}_{\frac{\rho_{\mathrm{B}}}{\beta \varrho }}  \exp \left( - \frac{\lambda \tau_{\mathrm{B}}d^{\mu} \sigma^2 }{\delta \rho(1 - \varrho)} \right) f_{P_{I}}(\rho) \mathrm{d}\rho\right]^2 
+  (1-\omega) W \! \int^{\frac{\rho_{\mathrm{B}}}{\beta \varrho } }_{0} \!  \exp \left( - \frac{\lambda \tau_{\mathrm{B}} d^{\mu} \sigma^2 }{\delta \rho(1 - \varrho)} \right) f_{P_{I}}(\rho) \mathrm{d}\rho 
 \nonumber \\ 
& \hspace{3mm} \times \! \int^{\infty}_{ \log_{2} (1+\tau_{\mathrm{H}})}  \int^{\infty}_{\frac{\rho_{\mathrm{H}}}{\beta\omega}}   \exp\left(\!- \frac{ \lambda d^{ \mu}\sigma^{2} (1-\omega) (2^t-1) }{ \omega \beta \rho - \rho_{\mathrm{H}} } \! \right)  
 \mathrm{Det}\big(\mathrm{Id}+\alpha \mathbb{C}_{\Psi}(\rho)\big)^{-\frac{1}{\alpha}} f_{P_{I}}(\rho) \mathrm{d}\rho \mathrm{d}t, 
 \label{eq:throughput_STP} 
\end{align} 
where  $f_{P_{I}}(\rho)$ has been obtained in $(\ref{eq:PDF})$  
and $\mathbb{C}_{\Psi} (\rho)$ is defined in
 (\ref{eqn:kernel_C}).

\end{theorem} 

\begin{proof}
By replacing $\mathcal{B}_{\mathrm{PTP}}$ in (\ref{eqn:average_throughput_PTP}) with $\mathcal{B}_{\mathrm{STP}}$ expressed as (\ref{eqn:delta_B_STP}), $\mathcal{T}_{\mathrm{STP}}$ can be obtained as in (\ref{eq:throughput_STP}).
\end{proof}

Though Theorems \ref{thm:throughput} and \ref{thm:throughput_STP} do not provide closed-form analytical expressions, the integrals can be efficiently evaluated by numerical analysis software like Matlab and Mathematica. Moreover, the expressions can be simplified considerably in some special cases like Corollary 
\ref{corollary2}. We only present the general results for the throughput expressions of PTP and STP in this paper due to limited space.

 
\section{Performance Evaluation and Analysis}
 
In this section, we validate our derived analytical expressions and conduct performance analysis based on numerical simulations. The performance of the proposed hybrid D2D communications is evaluated in the scenario coexisting with two groups of ambient transmitters $\Phi$ and $\Psi$, respectively, working on the energy harvesting frequency and active RF transmission frequency of the hybrid transmitter.
The transmit power level of the transmitters in $\Phi$ and $\Psi$ are set to be $P_{A}=P_{B}=0.2$ W, which is within the typical range of uplink transmit power for mobile devices. The interference ratio and transmission load are set to $\xi=0.2$ and $l_{A}=l_{B}=1$, respectively. The bandwidth of the transmitted signal $W$ in HTT mode is 1 MHz, and the noise variance $\sigma^2$ is -120 dBm/Hz. 
When the hybrid transmitter is in HTT mode, we assume equal time duration for energy harvesting and information transmission. In ambient backscattering mode, we consider 
$\rho_{\mathrm{B}} =8.9$ $\mu$W for
circuit power consumption  and $T_{\mathrm{B}}=$1 kbps for the transmission rate.

For the simulation of $\alpha$-GPP, we consider three typical scenarios, strong repulsion ($\alpha=-1$), medium repulsion ($\alpha=-0.5$) and no repulsion ($\alpha\to 0$, i.e., PPP), representing different social degrees among the ambient transmitters. 
In addition, for the evaluation of the Fredholm determinant, we adopt (\ref{evaluation_det}) and set $N_{\mathrm{closed}}$ to be 100. The other system parameters adopted in this section are listed in Table~\ref{parameter_setting} unless otherwise stated. 

\begin{table*} 
 \centering
 \caption{\footnotesize Parameter Setting.} \label{parameter_setting} 
 \begin{tabular}{|l|l|l|l|l|l|l|l|l|l|l|l|} 
 \hline
 Symbol & $\mu$ & $d$ & $R$ & $\theta$ & $\lambda$ & $\varrho$ & $\beta$ & $\delta$ & $\tau_{\mathrm{H}}$ & $\tau_{\mathrm{B}}$ & $\rho_{\mathrm{H}}$ \\ 
 \hline
 Value & 4 & 5 m & 30 m & 1 & 1 & 0.625 & 30 $\%$ & 1 & -40 dB & 5 dB & 113 $\mu$W \\
 \hline 
 \end{tabular}
 \end{table*} 

In the remaining of this section, the lines and symbols are used to represent the results evaluated from analytical expressions and Monte Carlo simulations, respectively. Additionally, for the comparison purpose, we evaluate the performance of a pure wireless-powered transmitter operated by the HTT protocol and a pure ambient backscatter transmitter as references, the plots of which are labeled as  ``Pure HTT" and ``Pure Ambient Backscattering", respectively. The performance of a pure wireless-powered transmitter (called pure HTT transmitter) and a pure ambient backscatter transmitter can be obtained by setting the hybrid transmitter in HTT mode and ambient backscattering mode, respectively, in all conditions. Specifically, the energy outage probability, coverage probability and average throughput of the pure ambient backscatter transmitter can be evaluated by $\mathcal{O}_{\mathrm{B}}$ in (\ref{CDF}), $\mathcal{C}_{\mathrm{B}}$ in (\ref{eqn:delta_B_PTP})  and $\mathcal{T}_{\mathrm{B}}$ in  (\ref{eqn:throughput_B}), respectively. Moreover, the energy outage probability, coverage probability and average throughput of the pure HTT transmitter can be evaluated by $\mathcal{O}_{\mathrm{H}}$ in (\ref{eqn:CDF_HTT}), $\mathcal{C}_{\mathrm{H}}$ in (\ref{eqn:coverageprobability_RSP_HTT}) and $\mathcal{T}_{\mathrm{H}}$ in (\ref{eqn:throughput_HTT}), respectively. 

\begin{figure}
\centering
 \begin{minipage}[c]{0.48\textwidth}
 \includegraphics[width=0.95\textwidth]{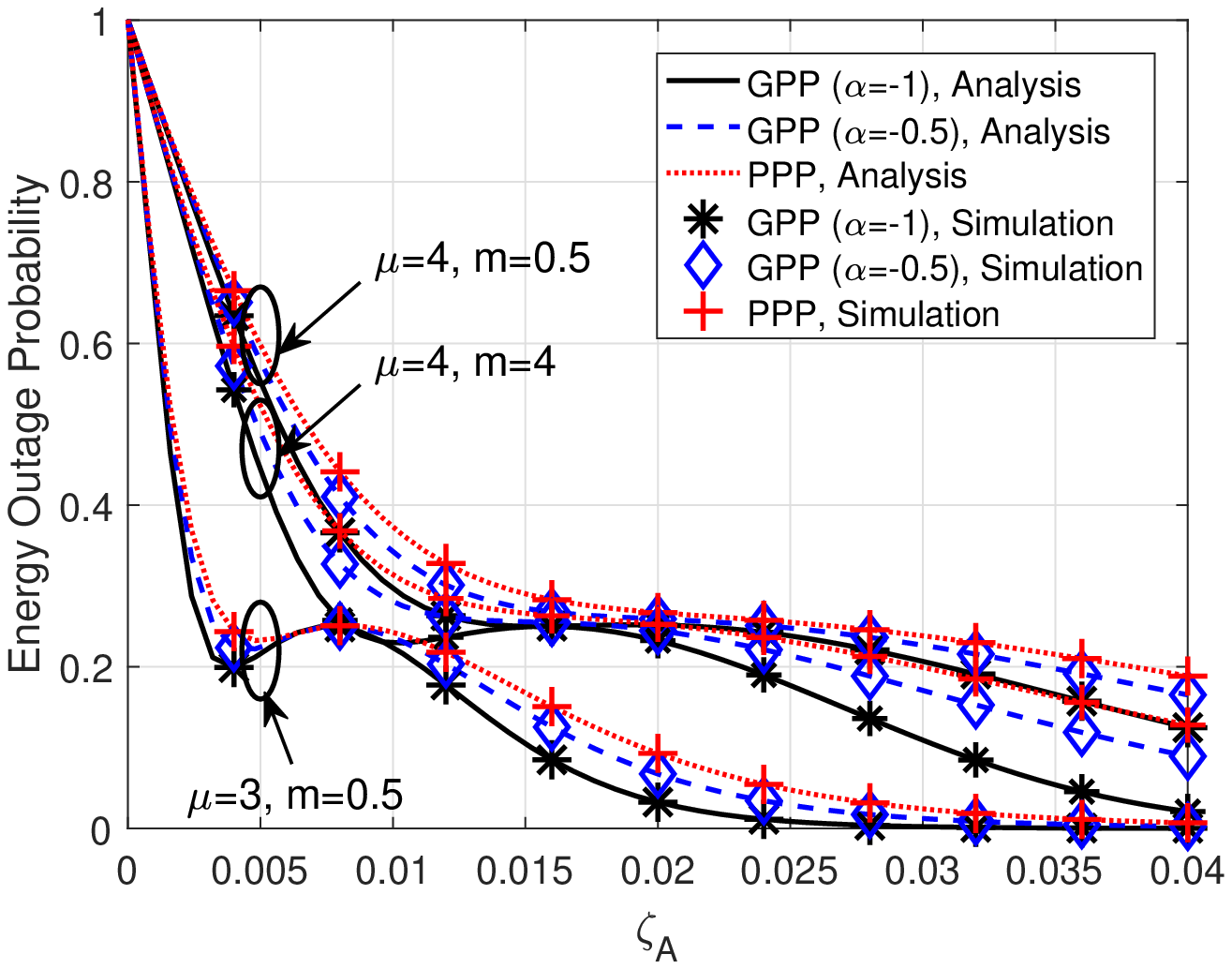} 
 \caption{
 	 $\mathcal{O}_{\mathrm{PTP}}$ as a function of $\zeta_{A}$.} \label{fig:EnergyOutage_Density_PTP}
 \end{minipage}
 \begin{minipage}[c]{0.48\textwidth}
 \includegraphics[width=0.95\textwidth]{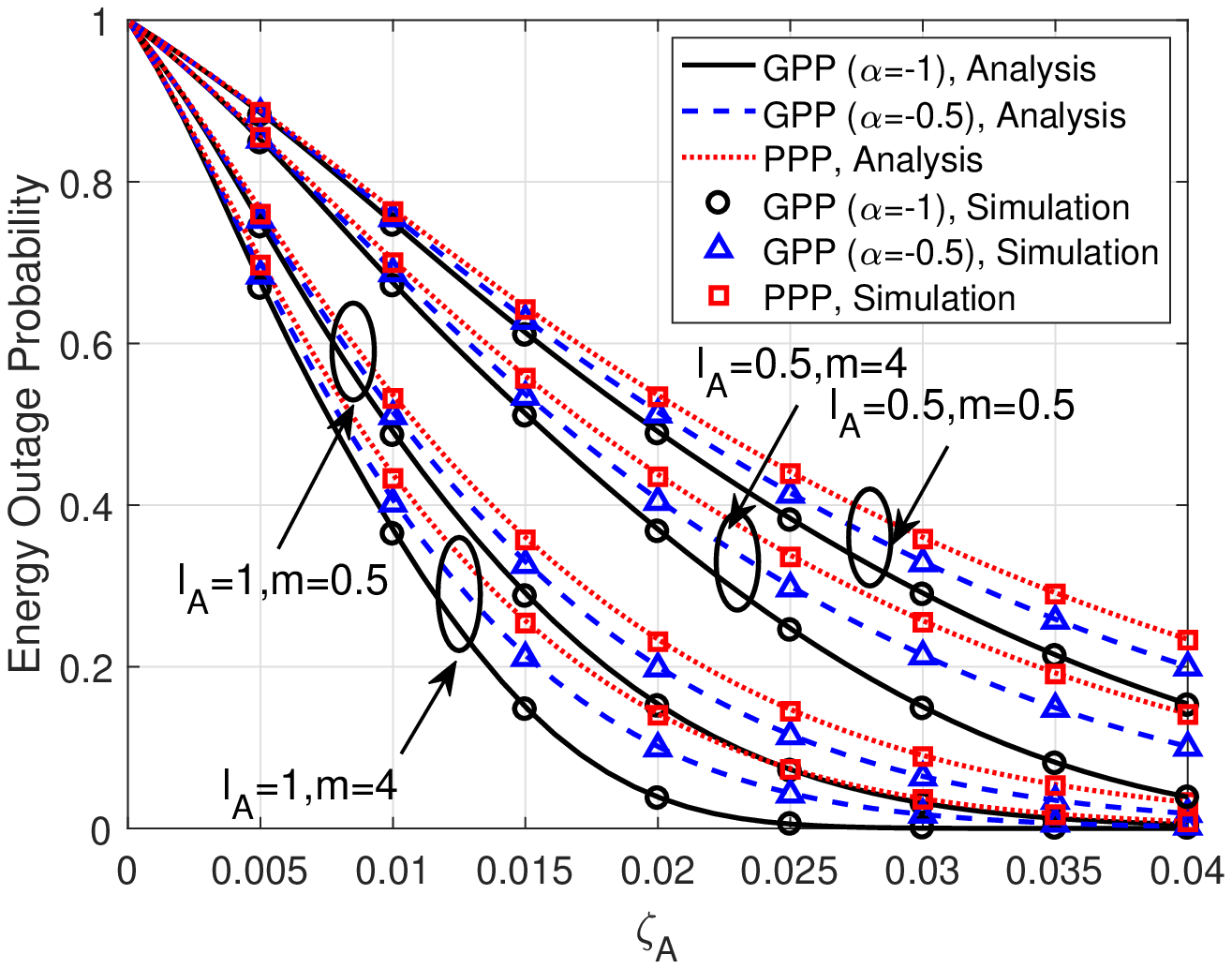}  
 \caption{
 	$\mathcal{O}_{\mathrm{STP}}$ as a function of $\zeta_{A}$.}\label{fig:EnergyOutage_Density_STP} 
 \end{minipage}
 \end{figure}

We first examine the energy outage probabilities. Figs.~\ref{fig:EnergyOutage_Density_PTP} and \ref{fig:EnergyOutage_Density_STP} show $\mathcal{O}_{\mathrm{PTP}}$ and $\mathcal{O}_{\mathrm{STP}}$ obtained in (\ref{eq:energy_outage_probability_PTP}) and (\ref{eq:energy_outage_probability_STP}), respectively, as a function of $\zeta_{A}$. Note that when $\zeta_{A}$ varies from 0 to 0.04, equivalently, the average number of ambient transmitters changes from 0 to 113. The accuracy of the energy outage probability expressions are validated by the simulation results with different values of $\alpha$ and $\mu$ under different transmission load $l_{A}$ and fading factors. In principle, larger $\zeta_{A}$ results in larger incident power at the hybrid transmitter, thus decreasing energy outage probabilities under a certain operation mode. However, one finds that only $\mathcal{O}_{\mathrm{STP}}$ is a monotonically decreasing function of $\zeta_{A}$ while $\mathcal{O}_{\mathrm{PTP}}$ not necessarily is. This is because the energy outage probability in HTT mode is higher than that in ambient backscattering mode given a certain $\zeta_{A}$. PTP works in ambient backscattering mode when $\zeta_{A}$ is low and $\mathcal{O}_{\mathrm{PTP}}$ first decreases with the increase of $\zeta_{A}$. When $\zeta_{A}$ reaches a certain level (e.g., 0.005 for case $\mu=3$), the hybrid transmitter is more in HTT mode, thus causing an increase of $\mathcal{O}_{\mathrm{PTP}}$. 
As for STP, it is in HTT mode when $\zeta_{A}$ is low. When $\zeta_{A}$ becomes higher, the STP is more in ambient backscattering mode, which means that lower energy outage probability can be achieved.
Therefore, mode switching results in a smooth and monotonic performance measure for $\mathcal{O}_{\mathrm{STP}}$.

From both Figs.~\ref{fig:EnergyOutage_Density_PTP} and \ref{fig:EnergyOutage_Density_STP}, we observe that the repulsion factor $\alpha$ among ambient transmitters has a considerable impact on energy outage probability. In other words, stronger attraction among the ambient transmitters leads to a higher energy outage probability of the hybrid transmitter. This can be understood that 
the incident power is more affected by the ambient transmitters in the vicinity of the hybrid transmitter. Strong repulsion generates a more scattered distribution of ambient transmitters guaranteeing that the hybrid transmitter is surrounded by ambient transmitters. By contrast, in the case of PPP, the distribution of ambient transmitters exhibits clustering behavior. Therefore, the likelihood that the hybrid transmitter has ambient transmitters nearby turns smaller, resulting in a higher chance of energy outage.  

 \begin{figure} 
 \centering
 \includegraphics[width=0.48\textwidth]{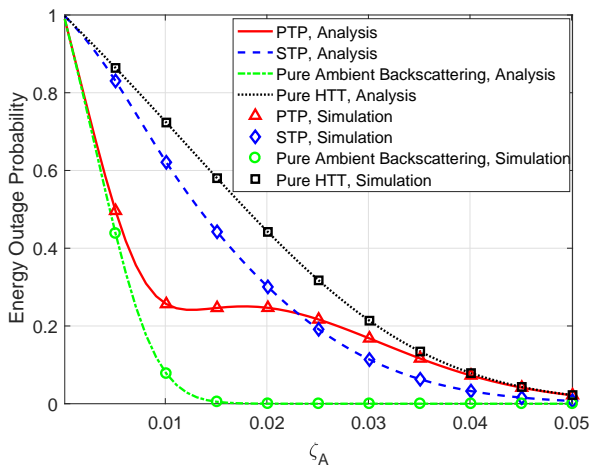}  
 \caption{Comparison of energy outage probabilities. ($\alpha=-1$)} \label{fig:EnergyOutageProb}
 \end{figure}

We observe that either a smaller path loss exponent (e.g., $\mu=3$ in Fig.~\ref{fig:EnergyOutage_Density_PTP}) or a larger Nakagami shape parameter $m$ (e.g., $m=4$ in Fig.~\ref{fig:EnergyOutage_Density_STP}) can reduce energy outage probabilities as both render less propagation attenuation. Additionally, as shown in Fig.~\ref{fig:EnergyOutage_Density_STP}, the transmission load $l_{A}$ is directly related to the aggregated energy harvesting rate, and thus the energy outage probability is inversely proportional to $l_{A}$.
 

Then, in Fig.  \ref{fig:EnergyOutageProb}, we compare energy outage probability of PTP, STP, pure ambient backscattering, and pure HTT 
under different ambient transmitter densities. 
It can be found that energy outage probabilities  are directly proportional to $\zeta_{A}$.
As expected, the pure ambient backscatter transmitter experiences less energy outage than the pure HTT transmitter in all cases due to lower circuit power consumption. Moreover, we observe that in terms of the energy outage probability, PTP is advantageous over STP when $\zeta_{A}$ is low (e.g., smaller than 0.02), and is outperformed by STP when $\zeta_{A}$ is high. This is due to the fact that PTP and STP, respectively, have better chance to be in  ambient backscattering and HTT modes if $\zeta_{A}$ is low, and tend to switch to the other mode otherwise.


 
Figs.~\ref{fig:CoverProb_density_PTP} and~\ref{fig:CoverProb_density_STP} illustrate how the coverage probabilities $\mathcal{C}_{\mathrm{PTP}}$ and $\mathcal{C}_{\mathrm{STP}}$ obtained in (\ref{eq:coverage_probability_Rayleigh_PTP}) and  (\ref{eq:coverage_probability_STP}), respectively, vary with ambient transmitter density $\zeta_{A}$ under different transmission loads and fading coefficients. In principle, larger density $\zeta_{A}$, repulsion factor $\alpha$, transmission load~$l_{A}$, and Nakagami shape parameter $m$ lead to more incident power, and thus, result in increased transmit power at the hybrid transmitter (either in ambient backscattering mode or in HTT mode) to improve the coverage probability. 
The mentioned effects on the coverage probability have been verified for both PTP and STP in Figs.~\ref{fig:CoverProb_density_PTP} and~\ref{fig:CoverProb_density_STP}, respectively, which indicates that both $\mathcal{C}_{\mathrm{PTP}}$ and $\mathcal{C}_{\mathrm{STP}}$ are monotonically increasing functions of $\zeta_{A}$, $\alpha$, $l_{A}$ and $m$. Note that from Figs.~\ref{fig:CoverProb_density_PTP}  and~\ref{fig:CoverProb_density_STP}, with the increase of $\zeta_{A}$, the coverage probabilities tend to be steady below 
1. This is because, given an interference ratio $\xi$, the increase of $\zeta_{A}$ not only provides the hybrid transmitter with more harvested energy to transmit, but also leads to more interference that harms the transmission.

\begin{figure} 
\centering
 \begin{minipage}[c]{0.48\textwidth}
 \includegraphics[width=0.999\textwidth]{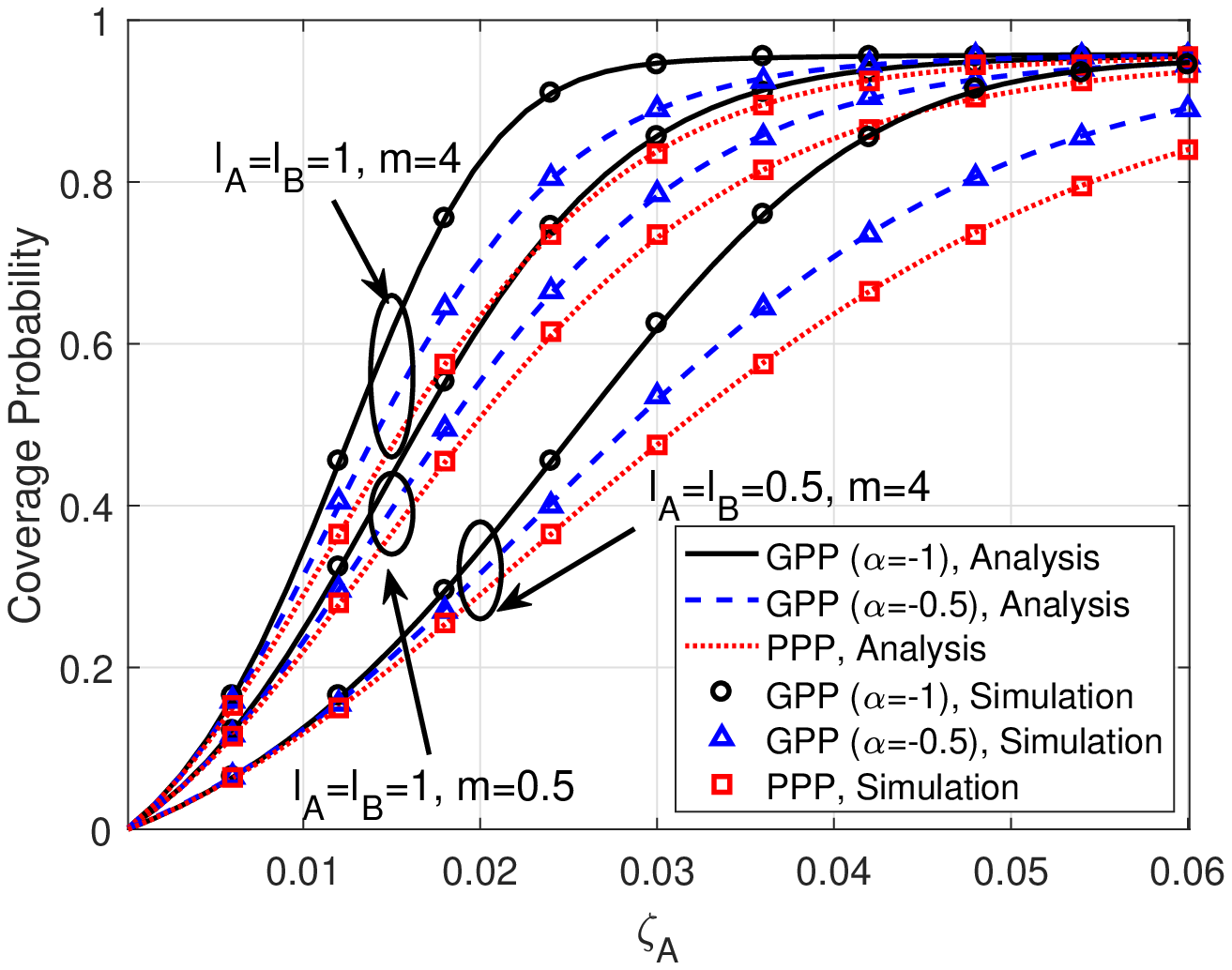} 
 \caption{$\mathcal{C}_{\mathrm{PTP}}$ as a function of $\zeta_{A}$. } \label{fig:CoverProb_density_PTP} 
 \end{minipage}
 \begin{minipage}[c]{0.48\textwidth}
 \includegraphics[width=0.999\textwidth]{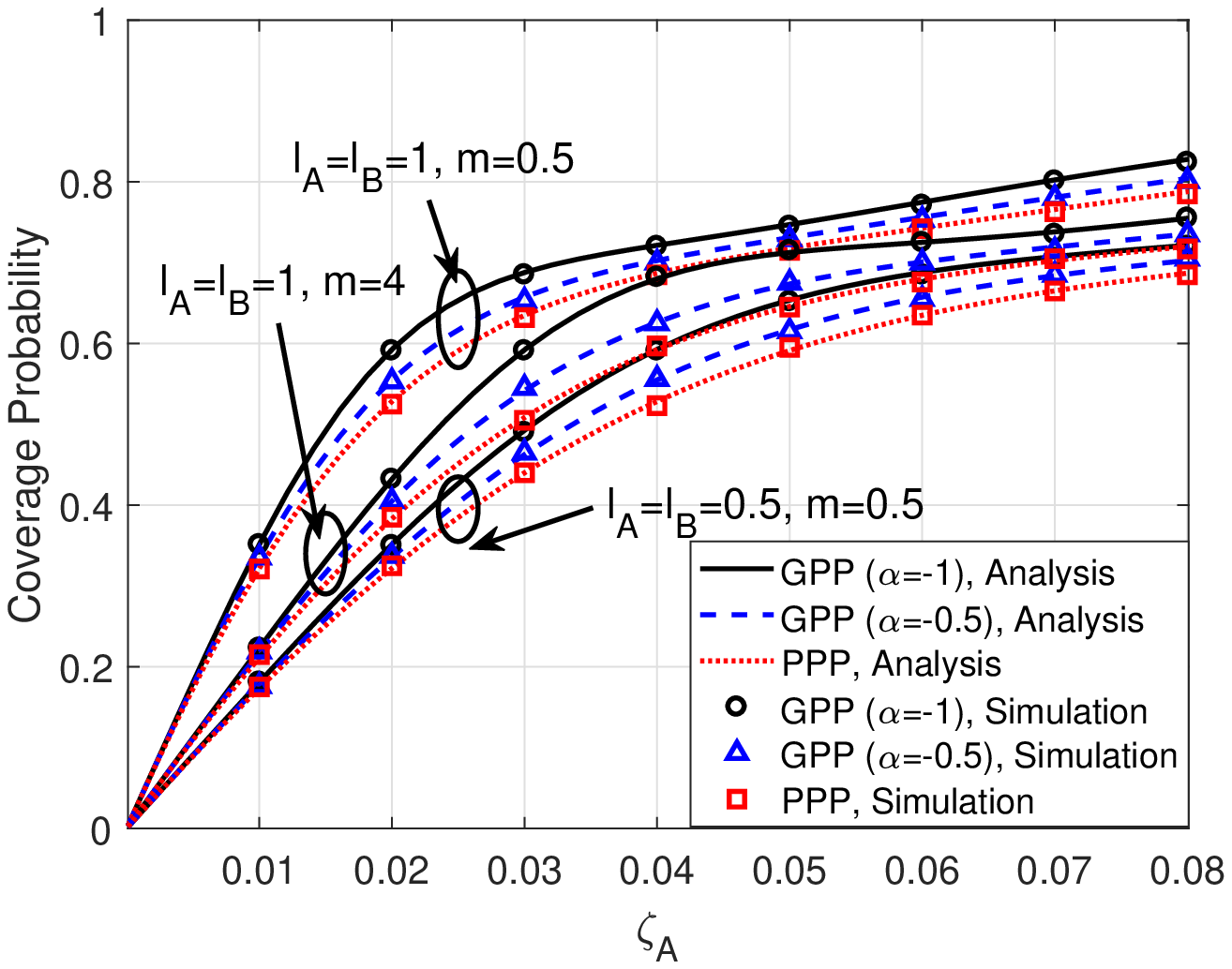} 
 \caption{$\mathcal{C}_{\mathrm{STP}}$ as a function of $\zeta_{A}$. }\label{fig:CoverProb_density_STP} 
 \end{minipage}
 \end{figure}


Fig.~\ref{CP_density_comparison} compares coverage probabilities (as functions of density $\zeta_A$) of PTP, STP, pure ambient backscattering, and pure HTT. 
When $\xi$ is small (i.e., $\xi=0.2$) as shown in Fig.~\ref{fig:CP_density_comparison1}, the pure HTT transmitter experiences low interference, and thus, achieves significantly higher coverage probability than pure ambient backscattering.
However, in the case with high interference ratio (i.e., $\xi=0.8$) as depicted in Fig.~\ref{fig:CP_density_comparison2}, their performance gap becomes smaller and pure ambient backscattering outperforms pure HTT when $\zeta_{A}$ is large (e.g., above 0.06), due to the high interference received by the pure HTT transmitter. 
We also observe that PTP achieves similar performance to that of STP under small $\zeta_{A}$ and is obviously outperformed by STP as $\zeta_{A}$ grows larger (e.g., above 0.06). The reason behind is that PTP selects operation mode solely based on the incident power and is unaware of the interference level so that it remains in HTT mode even when the achieved SINR is low. This reflects that STP is more suitable for use in an interference rich environment. 

\begin{figure} 
\centering
 \subfigure [\vspace{-10mm}]
  {
\label{fig:CP_density_comparison1}
 \centering   
 \includegraphics[width=0.48 \textwidth]{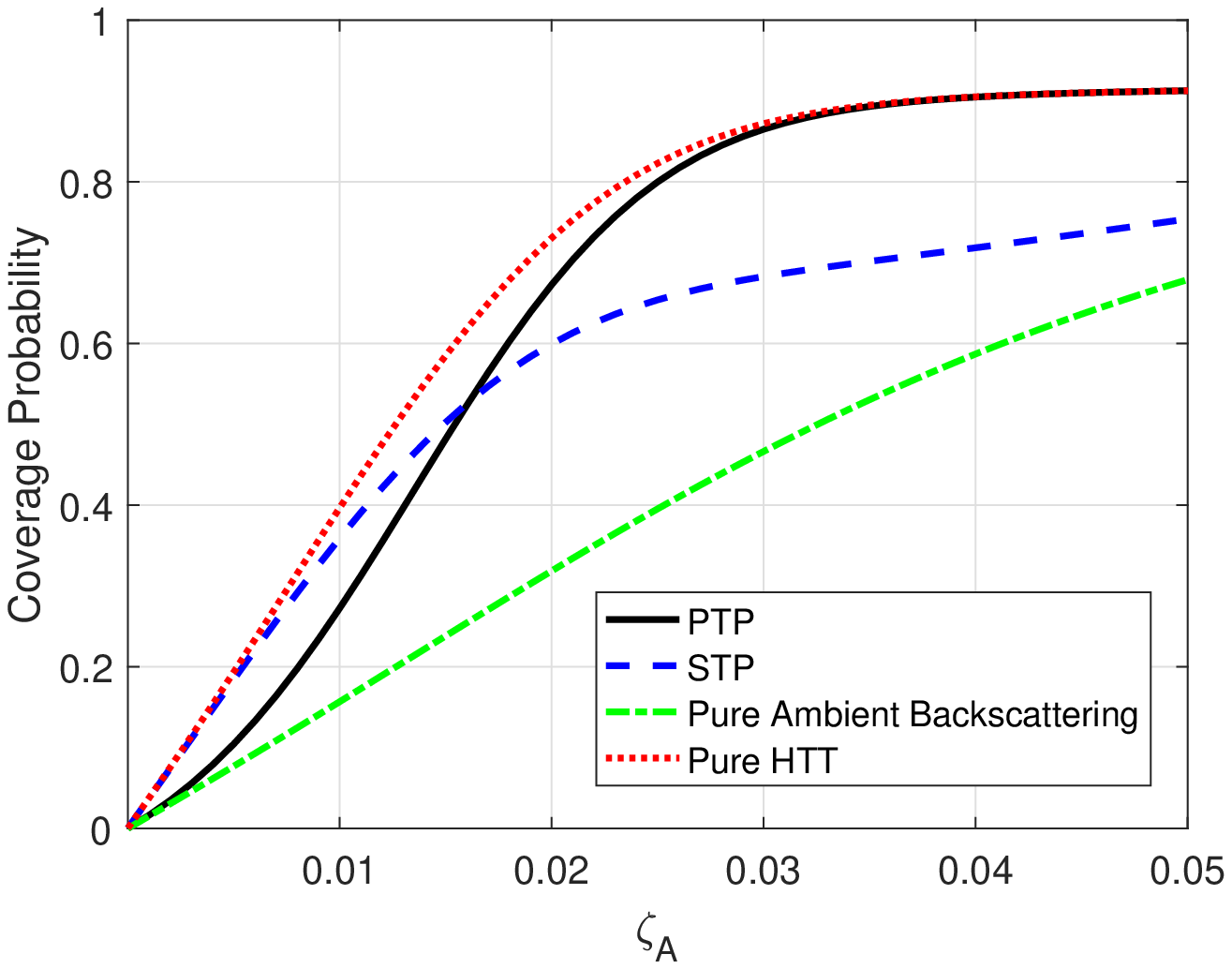}}
 \centering  
 \subfigure  [ 
 ] {
\label{fig:CP_density_comparison2}
 \centering
\includegraphics[width=0.48 \textwidth]{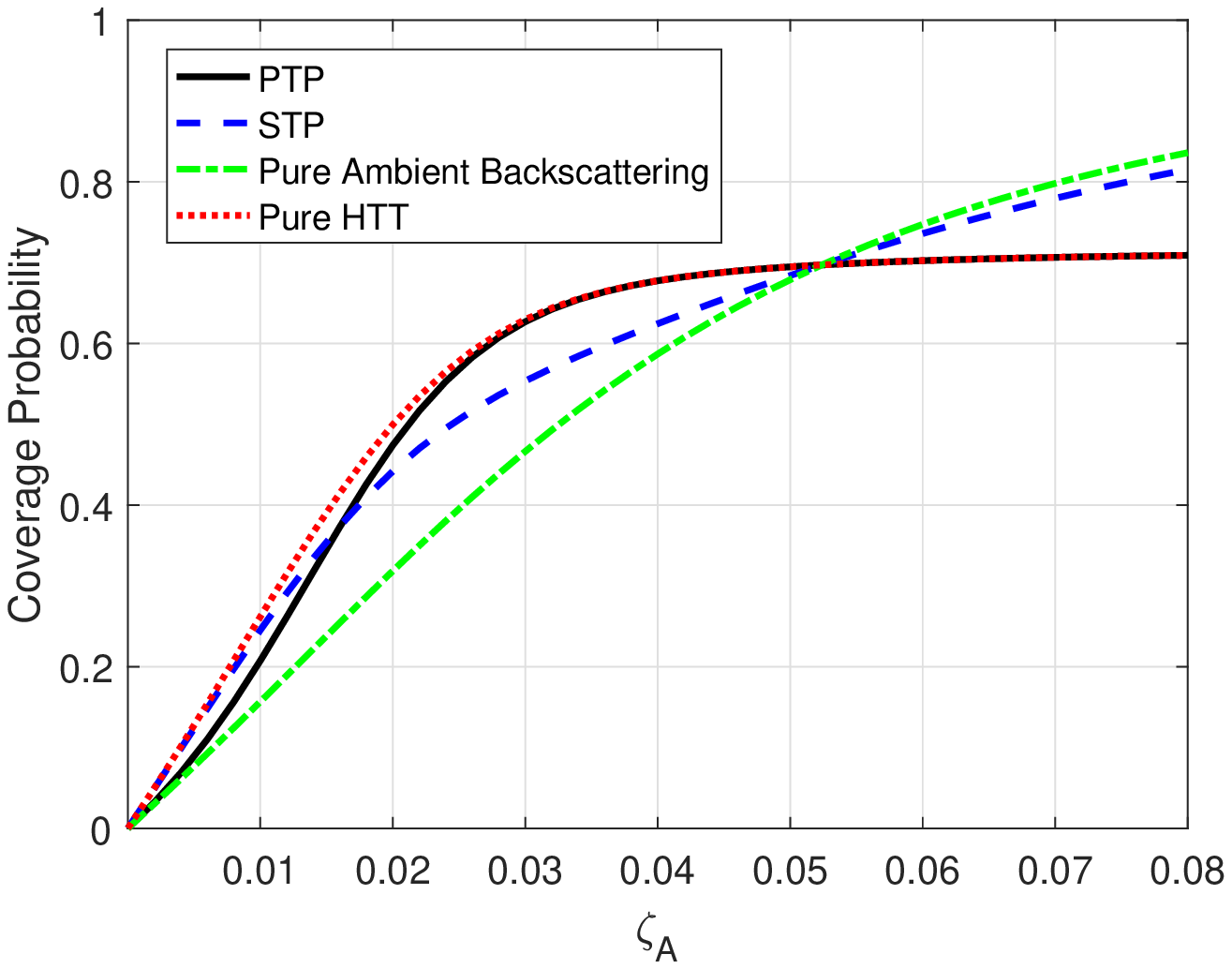}}
\caption{Comparison of coverage probabilities as a function of $\zeta_{A}$. ((a) $\xi=0.2$, (b) $\xi=0.8$) } 
\centering
\label{CP_density_comparison}
\end{figure}

In Fig. \ref{fig:_CP_BE}, we show the coverage probability as a function of backscattering efficiency $\delta$ when $\zeta_{A}$ is set at 0.02 and 0.04. As pure HTT is not affected by the backscattering efficiency, the resulting coverage probability remains constant. We observe that the coverage probability of a pure backscattering transmitter is a monotonically increasing function of the backscattering efficiency. 
Under PTP, when $\zeta_{A}$ is small (e.g., $\zeta_{A}=0.02$), the hybrid transmitter is likely to select either HTT mode or ambient backscattering mode, resulting in a coverage probability between that of pure HTT and that of pure ambient backscattering. When $\zeta_{A}$ is large (e.g., $\zeta_{A}=0.04$), the hybrid transmitter has very high chance to stay in HTT mode, and thus results in a coverage probability almost overlapping with that of pure HTT.   
Under STP, when $\zeta_{A}=0.02$, the increase of backscatter efficiency gives the hybrid transmitter more chance to select ambient backscattering mode which has lower coverage probability than that of HTT mode, and therefore, the overall coverage probability of STP decreases. When $\zeta_{A}=0.04$, the hybrid transmitter also has larger chance to select ambient backscattering mode as the backscattering efficiency increases. However, in this case, the coverage probability of ambient backscattering mode is significantly improved with higher backscattering efficiency. Thus, the overall coverage probability of STP increases with $\zeta_{A}$.

\begin{figure}
\centering
 \begin{minipage}[c]{0.48\textwidth} 
 \includegraphics[width=1.1\textwidth]{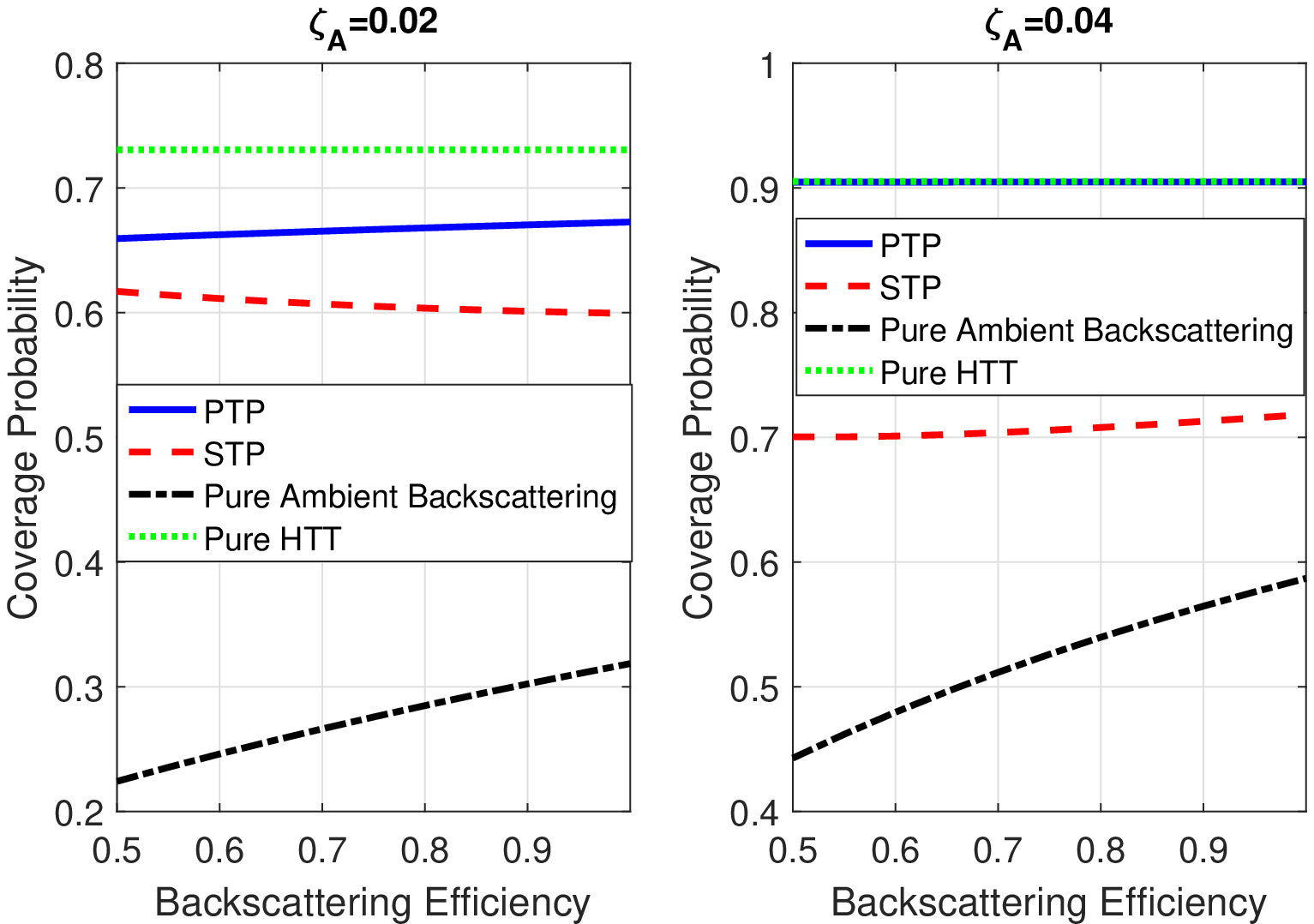} 
 \caption{Coverage probability as a function of backscattering efficiency.} \label{fig:_CP_BE}
 \end{minipage}
 \begin{minipage}[c]{0.48\textwidth} 
 \includegraphics[width=1.1\textwidth]{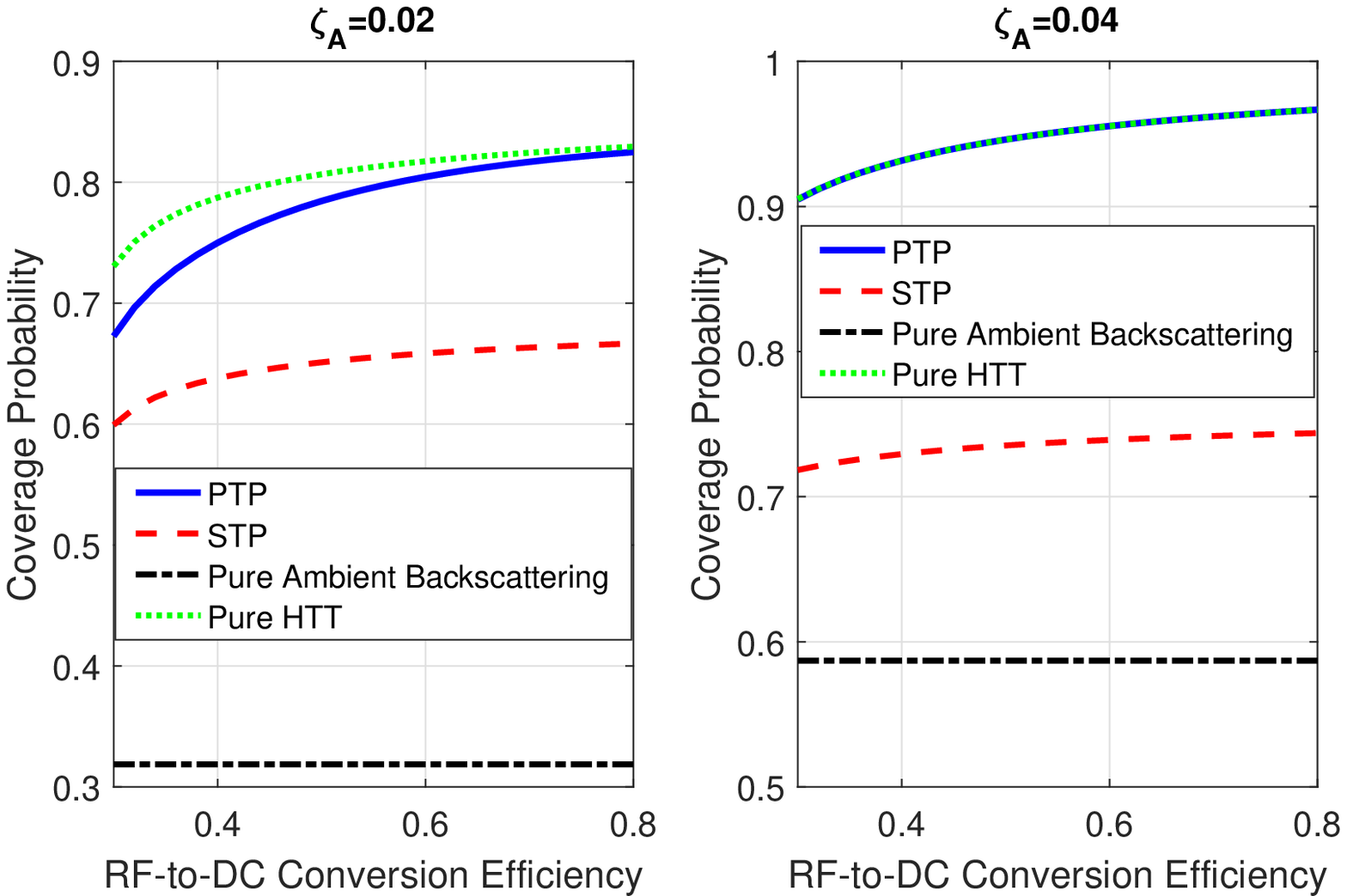}  
 \caption{Coverage probability as a function of RF-to-DC conversion efficiency.}\label{fig:CP_beta}
 \end{minipage}
 \end{figure}

In Fig.~\ref{fig:CP_beta}, we demonstrate how the coverage probabilities vary with the RF-to-DC conversion efficiency $\beta$ when $\zeta_{A}$ is set at 0.02 and 0.04. It is straightforward that the coverage probabilities are monotonically increasing functions of $\beta$. We can also see that the variations of the coverage probabilities due to the change of $\beta$ decrease as $\zeta_{A}$ becomes larger. This indicates that higher RF-to-DC conversion efficiency is more beneficial to the coverage probability of the hybrid transmitter when the density of ambient transmitters is small.  Additionally, the coverage probability of a pure backscattering transmitter changes very slightly as $\beta$ varies. 
This is because the coverage probability is mainly affected by two factors, i.e., energy harvesting rate $P^{\mathrm{B}}_{E} =\beta \varrho P_{I}$ and effective backscattered power $P_{R} = \delta (1-\varrho) P_{I}  $. Once the energy harvesting rate exceeds the circuit power consumption of a pure backscattering transmitter $\rho_{\mathrm{B}}$, the effective backscattered power is not impacted by the energy harvesting rate. Due to the fact that $\rho_{\mathrm{B}}$ is very small, the energy harvesting rate reaches $\rho_{\mathrm{B}}$ with a probability approaching 1 at both $\zeta_{A}=0.02$ and $\zeta_{A}=0.04$. Therefore, the variation of $\beta$ within a normal range, i.e., from 0.3 to 0.8, does not cause significant change on the coverage probability of a pure backscattering transmitter.

Furthermore, Fig.~\ref{CP_d_comparison} illustrates the comparisons of the coverage probabilities (as functions of transmitter-receiver distance $d$) under different density of the ambient transmitters $\zeta_{A}$ and interference ratio $\xi$. 
We focus on evaluating the scenario with both small $\zeta_{A}$ and $\xi$ and the scenario with both large $\zeta_{A}$ and $\xi$.\footnote{The coverage probabilities of the hybrid transmitter are increasing functions of $\zeta_{A}$, as larger density of ambient transmitters produces more RF signals for the hybrid transmitter to perform either active transmission or ambient backscattering. Moreover, the coverage probabilities of the hybrid transmitter are decreasing functions of $\xi$, as larger interference results in lower received SINR at the hybrid receiver. Therefore, it is straightforward that the coverage probabilities of the hybrid transmitter are higher in the scenario with both smaller $\xi$ and larger $\zeta_{A}$, and become smaller with both larger $\xi$ and smaller $\zeta_{A}$. We omit showing the above two scenarios due to the space limit.} 
In the former scenario 
(i.e., $\zeta_{A}=0.02$ and $\xi=0.1$) as shown in Fig.~\ref{fig:EnergyOutage_d_comparison1}, the pure HTT transmitter is inferior to the pure ambient backscatter transmitter when $d$ is small (e.g., $d < 2$). It is because the pure HTT transmitter has a higher chance of energy outage when $\zeta_{A}$ is small. However, the pure HTT transmitter is more robust to longer $d$ since it first aggregates the harvested energy and generates higher transmit power than backscattered power. Moreover, with the increase of $d$ from 0, STP first outperforms PTP by operating in ambient backscattering mode in low $\zeta_{A}$ and is outperformed by PTP when $d$ is larger due to the same reason. Eventually, both achieve comparable performance when $d$ is above a certain value (i.e., around 7 m). 
\begin{figure}  
\centering 
 \subfigure [
 ]  {
\label{fig:EnergyOutage_d_comparison1}
 \centering 
 \includegraphics[width=0.48 \textwidth]{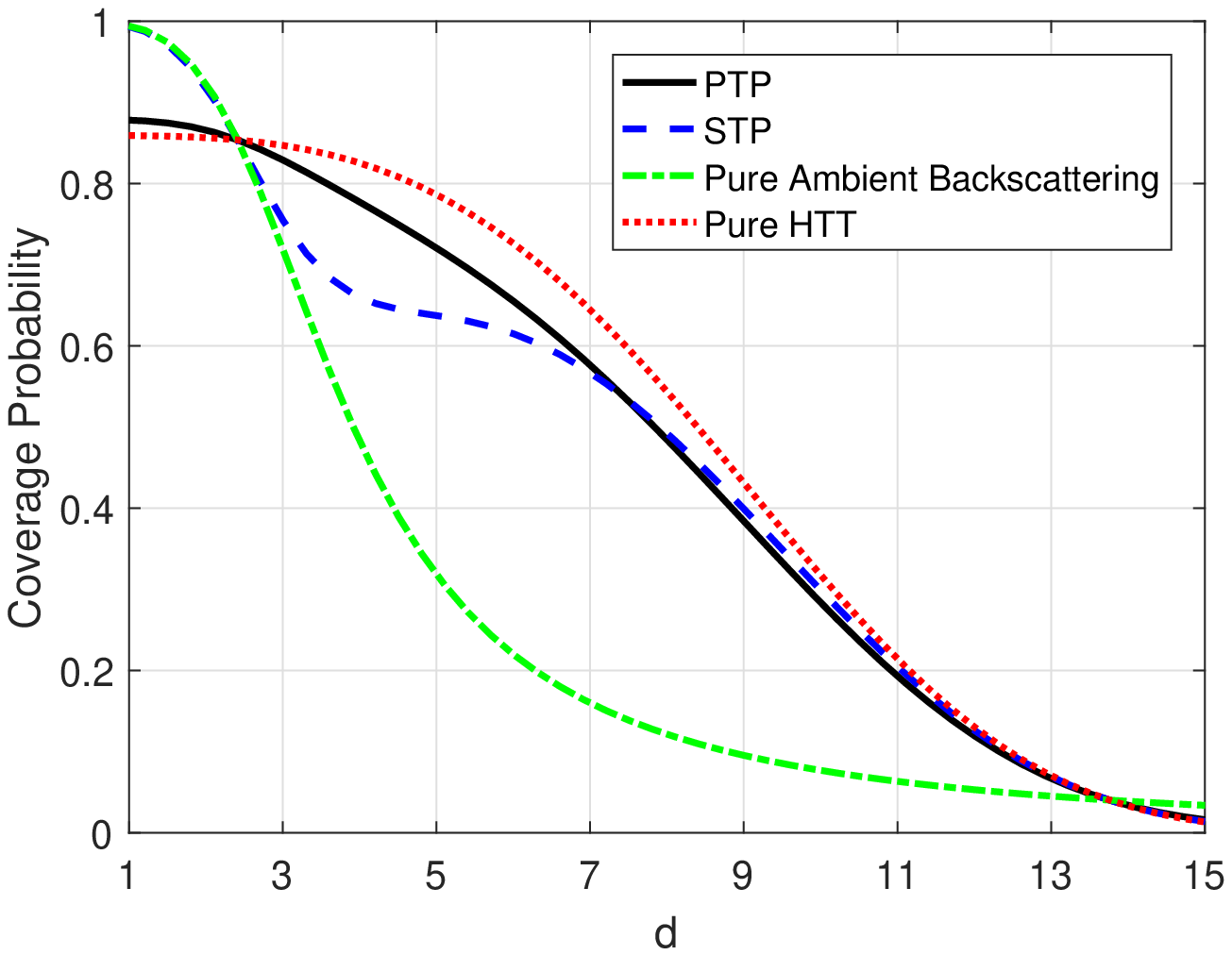}} 
 \centering 	
 \subfigure [ 
 ] {
\label{fig:EnergyOutage_d_comparison2}
 \centering
\includegraphics[width=0.48 \textwidth]{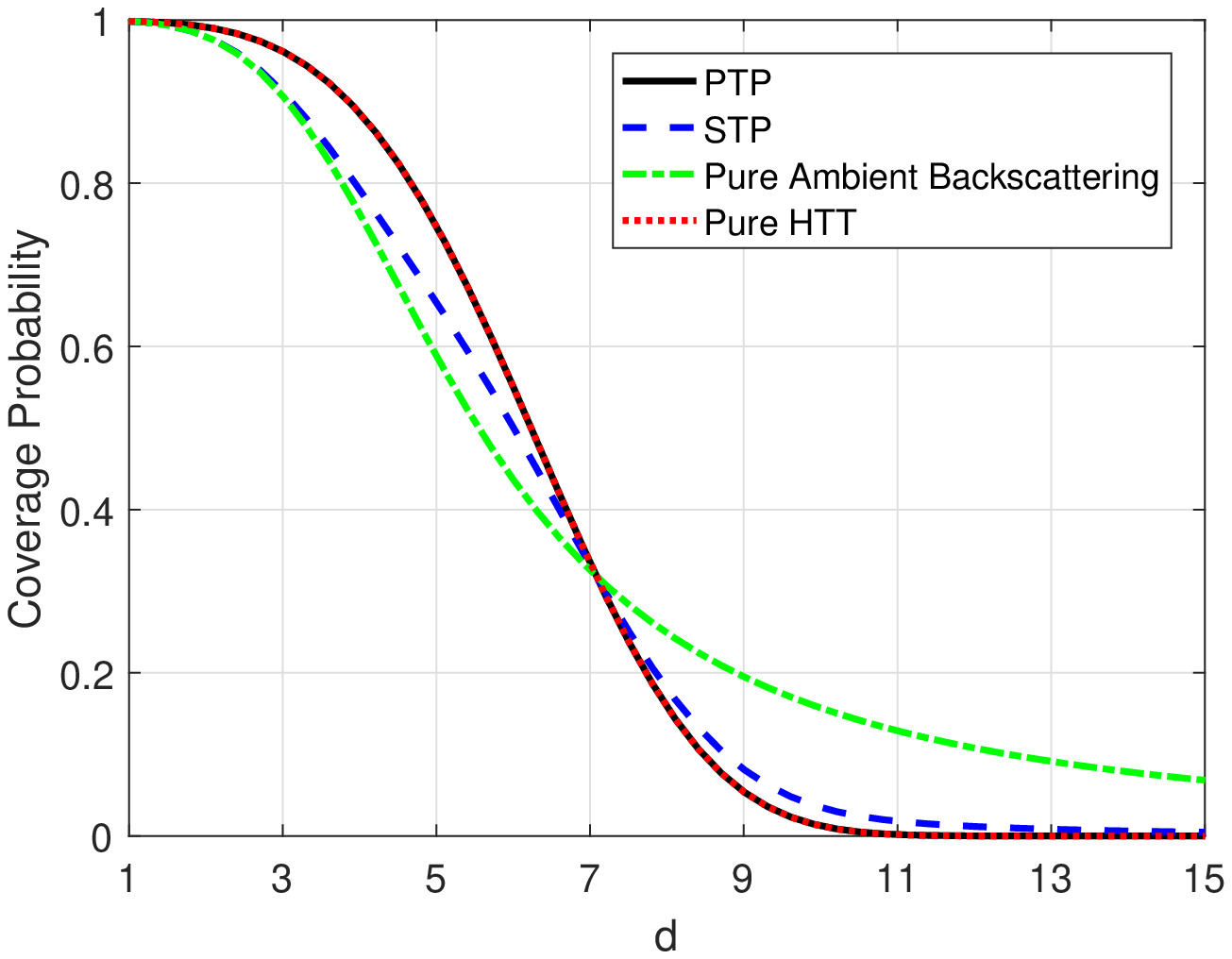}}
\caption{Comparison of coverage probabilities as a function of $d$. ((a) $\zeta_{A}=0.02$, $\xi=0.1$, (b) $\zeta_{A}=0.04$, $\xi=0.6$)} 
\centering
\label{CP_d_comparison}
\end{figure}
Conversely, in the scenario with 
larger $\zeta_{A}$ and $\xi$ (i.e., $\zeta_{A}=0.04$ and $\xi=0.6$), as depicted in Fig.~\ref{fig:EnergyOutage_d_comparison2}, the pure HTT transmitter is superior to the pure ambient backscatter transmitter when $d$ is small (e.g., $d < 6$) because abundant ambient RF resources mitigate the occurrence of energy outage. However, due to severe interference, $\mathcal{C}_{\mathrm{H}}$ (coverage probability of the pure HTT transmitter) plunges with the increase of $d$. Instead, the pure ambient backscatter transmitter becomes more robust to longer $d$. It can be seen that $\mathcal{C}_{\mathrm{PTP}}$ overlaps with $\mathcal{C}_{\mathrm{H}}$ because when the harvested energy is ample the hybrid transmitter always operates in HTT mode. Overall, the performance gap between PTP and STP is small in this scenario. 

\begin{figure} 
\centering 
 \subfigure [ 
 ] {
\label{fig:throughput_xi02}
 \centering
 \includegraphics[width=0.48 \textwidth]{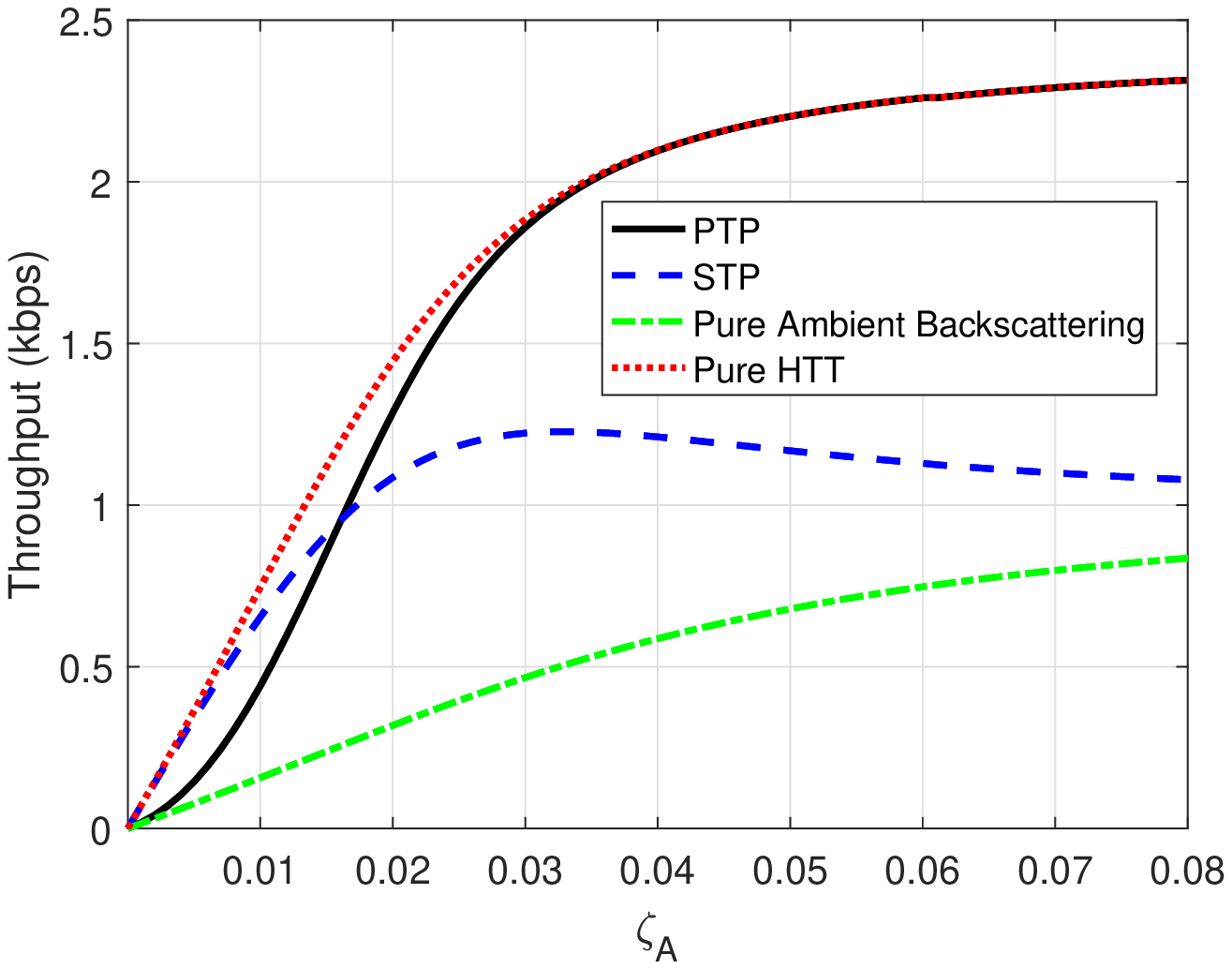}}
 \centering  
 \subfigure [
 ] {
\label{fig:throughput_xi08}
 \centering
\includegraphics[width=0.48 \textwidth]{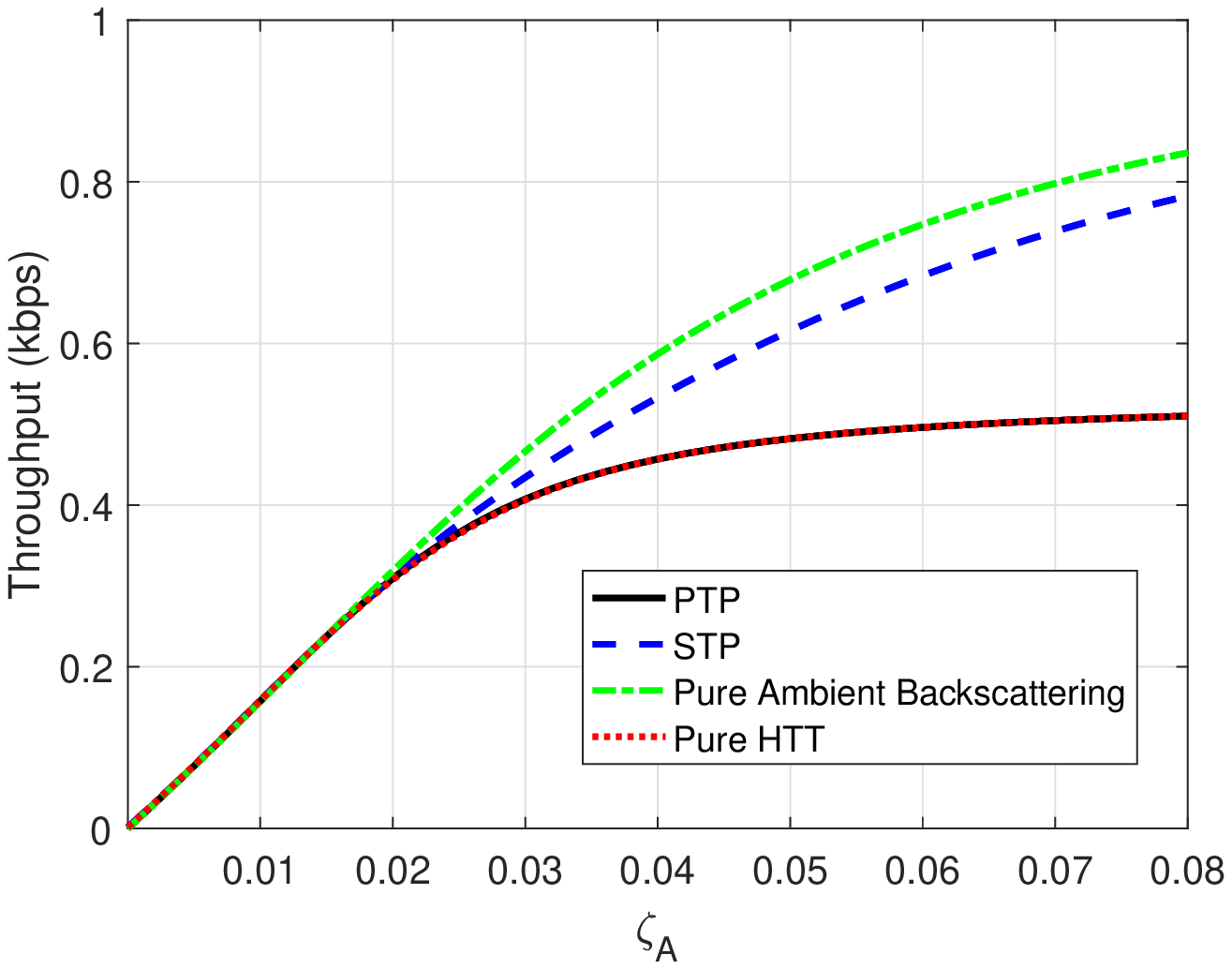}}\\
\caption{Comparison of average throughput as a function of $\zeta_{A}$. ((a) $d=5$, $\xi=0.2$, (b) $d=5$, $\xi=0.8$)}  
\centering
\label{Throughput_density_comparison}
\end{figure}


Fig.~\ref{Throughput_density_comparison} compares the throughput (as a function of density $\zeta_A$) of PTP, STP, pure ambient backscattering, and pure HTT.
We focus on the cases when the interference ratio is small ($\xi=0.2$) and large ($\xi=0.8$) with the corresponding results shown in Figs.~\ref{fig:throughput_xi02} and~\ref{fig:throughput_xi08}, respectively. We observe that the trend of throughput performance has been somehow reflected by the coverage probabilities shown in Fig.~\ref{CP_density_comparison}. Similar to our observation for Fig.~\ref{CP_density_comparison}, we can draw the conclusion that, in general, PTP yields higher throughput when the interference level is low. Otherwise, STP is more suitable for use.

Additionally, Fig.~\ref{Throughput_d_comparison} examines the influence of transmitter-receiver distance $d$ on the throughput performance. As expected, the pure HTT transmitter prominently outperforms the pure ambient backscatter transmitter with relatively smaller $\xi$ and larger $\zeta_{A}$ (i.e., $\xi=0.2$ and $\zeta_{A}=0.02$) as shown in Fig.~\ref{fig:D002X02}. The performance gap becomes progressively significant with decreasing $d$. PTP also attains remarkable throughput gain over STP since, in this context, energy harvesting rate is a better indication to select HTT mode. By contrast, in the case with relatively larger $\xi$ and smaller $\zeta_{A}$ (i.e., $\xi=0.8$ and $\zeta_{A}=0.01$) as shown in Fig.~\ref{fig:D001X08}, PTP is less advantageous than STP because the energy harvesting rate detection in PTP fails to take into account the increased interference. However, PTP becomes superior and exhibits less susceptibility as $d$ grows. The cause is that when $\zeta_{A}$ is small, PTP operates more in ambient backscattering mode which is immune to the increased interference.
 

\begin{figure} 
\centering
 \subfigure [   ] {
\label{fig:D002X02}
 \centering
 \includegraphics[width=0.48 \textwidth]{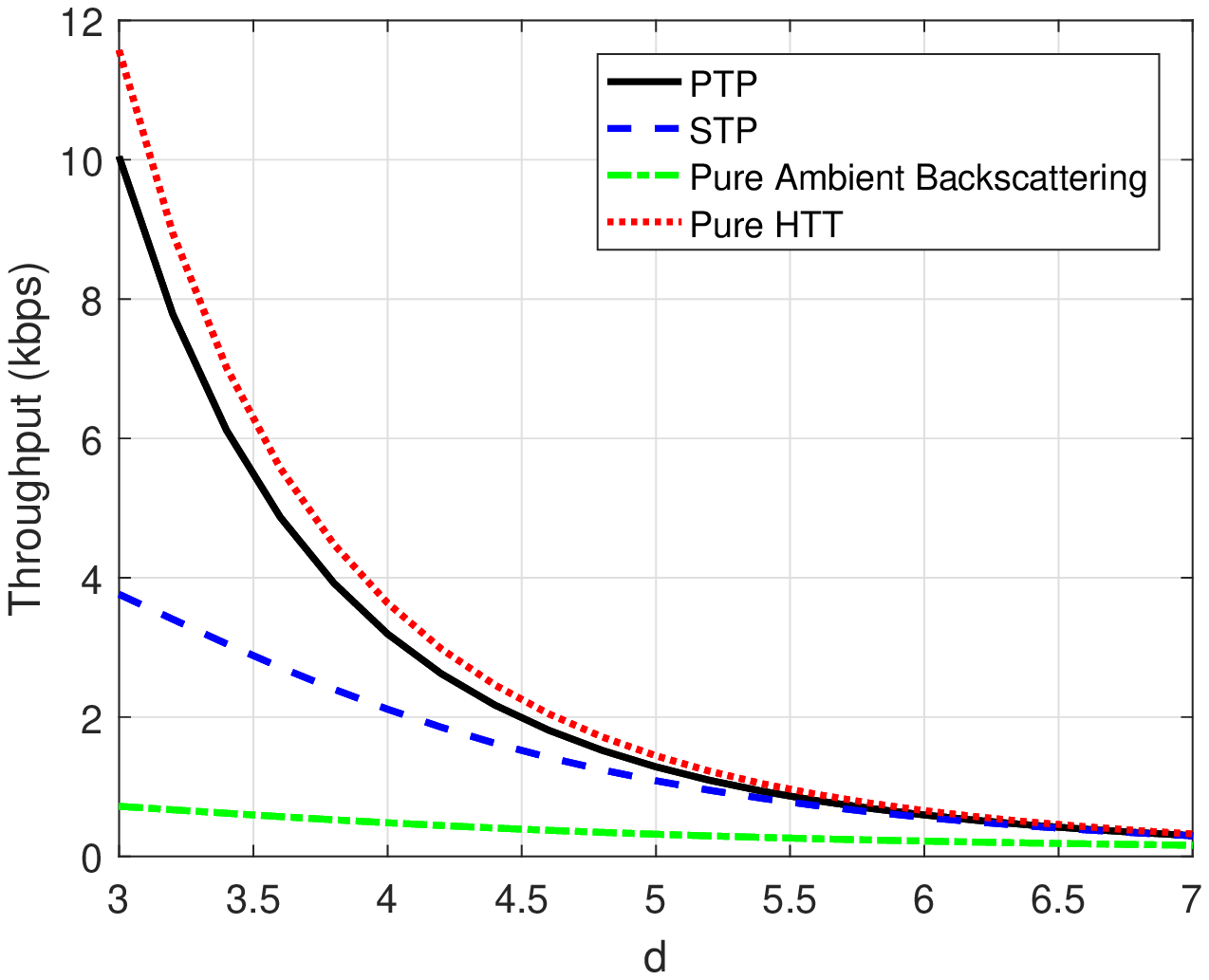}}
 \centering
 \subfigure [ ] {
\label{fig:D001X08}
 \centering
\includegraphics[width=0.48 \textwidth]{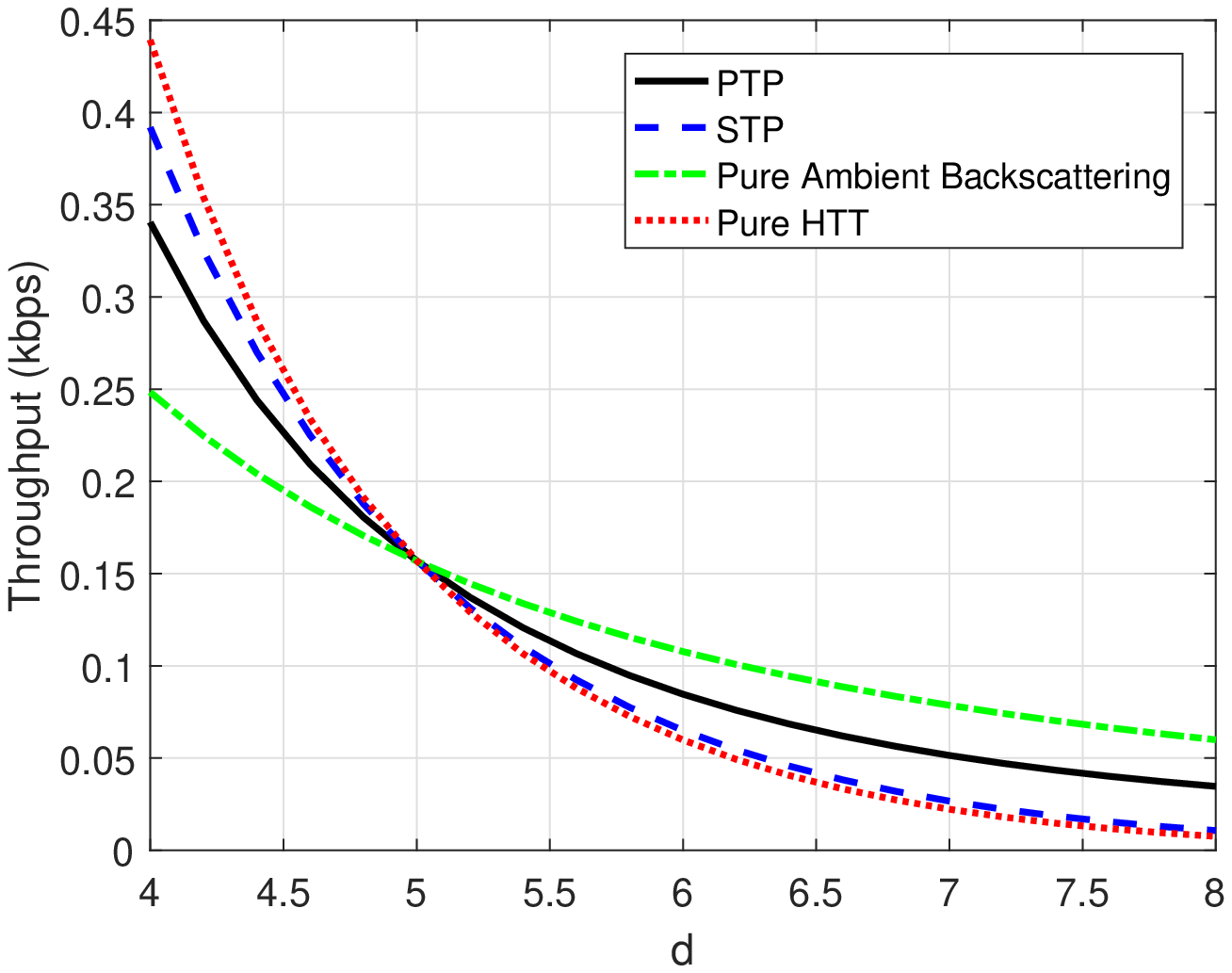}}\\
\caption{Comparison of average throughput as a function of $d$. ((a) $\xi=0.2$, $\zeta_{A}=0.02$, (b) $\xi=0.8$, $\zeta_{A}=0.01$)}  
\centering
\label{Throughput_d_comparison} 
\end{figure}

 
\section{Conclusion and Future Work}

In this paper, we have introduced a novel paradigm of hybrid D2D communications that integrate ambient backscattering with wireless-powered communications. To enable the operation of our proposed hybrid transmitter in diverse environments, two simple mode selection protocols, namely PTP and STP, have been devised based on the energy harvesting rate and received SNR of the modulated backscatter, respectively. Under the framework of repulsive point process modeling, we have analyzed the hybrid D2D communications and focused on investigating the impact of environment factors. In particular, the performance of the hybrid D2D communications has been characterized in terms of energy outage probability, coverage probability, and average throughput. The performance analysis has shown that the self-sustainable D2D communications benefit from larger repulsion, transmission load and density of ambient energy sources. Moreover, we have found that PTP is more suitable for use in the scenarios with a large density of ambient energy sources and low interference level. On the contrary, STP becomes favorable in the scenarios when the interference level and density of ambient energy sources are both low or both high. Additionally, PTP appears to be more reliable to yield better throughput for long-range transmission in general. 


The performance of our proposed hybrid transmitter and receiver can be improved when
multiple antennas are adopted. However, it is challenging to characterize the hybrid D2D communication
performance, as a physical-layer multiple-input multiple-output (MIMO) channel model for ambient backscatter
transmitter and receiver is not available yet. Especially, how multiple antennas at an ambient backscatter transmitter can facilitate load modulation and how multiple antennas at an
ambient backscatter receiver can render signal detection are open research issues. When such a MIMO channel model is available, an intriguing future direction is to extend our analytical framework to model the performance of the hybrid D2D communications in various cases of MIMO channels.

Another future work is to investigate the hybrid D2D communications in the scenarios where there exist randomly located hybrid transmitter-receiver pairs.
A possible research topic is to design distributed mode selection protocols performed by individual hybrid transmitters based on their local information, e.g., available harvested energy and channel state information. It is also interesting to design centralized mode selection protocols. 

\section*{Appendix I} \label{proof_theorem1}
\begin{proof} 
	The distribution of the aggregated received power at the origin from ambient transmitters can be determined by the calculation of its Laplace transform. Specifically, the Laplace transform of the accumulated incident power at the antenna of the hybrid transmitter can be obtained as
	\begin{align} 
	\mathcal{L}_{P_{I}}(s) & = \mathbb{E} \left[ \exp\left(-s P_{I} \right) \right] \nonumber 
	 = \! \mathbb{E} \left[ \prod_{a \in \mathcal{A}} \exp\left(- \frac{s P_{A} h_{a,\mathrm{S}}}{ \|\mathbf{x}_{a}-\mathbf{x}_{\mathrm{S}}\|^{\mu}} \right) \right]  
	 \nonumber \\ 	&
	 = \mathbb{E} \left[ \prod_{a \in \mathcal{A}} M_{h} \left( - \frac{s P_{A}}{ \|\mathbf{x}_{a}-\mathbf{x}_{\mathrm{S}}\|^{\mu}} \right) \right] \nonumber \\ 	&
	= \mathbb{E} \Bigg[ \exp \left ( \sum_{a \in \mathcal{A}} \ln \left ( M_{h} \left( - \frac{s P_{A}}{ \|\mathbf{x}_{a}-\mathbf{x}_{\mathrm{S}}\|^{\mu}} \right) \right) \right) \Bigg] 
	\nonumber \\ 	&
	\overset{\text{(i)}}{=} \mathrm{Det}\big(\mathrm{Id}+\alpha \mathbb{A}_{\Phi}(s)\big)^{-\frac{1}{\alpha} }, \label{LP_incident_rate}
	\end{align}	
	where $M_h(\cdot)$ is the MGF of $h_{a,\mathrm{S}}$ and $(\mathrm{i})$ follows by applying Proposition~\ref{lemma1}, and $\mathbb{A}_{\Phi}$  is 
	\begin{align}
	\label{eq:kernal_A}
	& \mathbb{A}_{\Phi}(s)=\sqrt{1-M_{h}\left(- s P_{A} \|\mathbf{x}-\mathbf{x}_{\mathrm{S}}\|^{-\mu} \right)} \nonumber \\  & \hspace{15mm} \times G_{\Phi}(\mathbf{x},\mathbf{y}) \sqrt{1-M_{h}\left(- s P_{A} \|\mathbf{y}-\mathbf{x}_{\mathrm{S}}\|^{-\mu} \right)},
	\end{align}
where $G_{\Phi}$ is the Ginibre kernel given in (\ref{eq:ginibre}). Since $h \sim \mathcal{G}(m,\frac{\theta}{m})$, the MGF of a Gamma random variable $h$ can be calculated as $M_{h}(z)=(1-\frac{\theta z}{m})^{-m}$. Therefore, we have 
	\begin{eqnarray} \label{MPG}
	M_{h}\left(- s P_{A} \|\mathbf{x}-\mathbf{x}_{\mathrm{S}}\|^{-\mu} \right)= \left( 1+ \frac{s\theta P_{A} }{m\|\mathbf{x}-\mathbf{x}_{\mathrm{S}}\|^{\mu}} \right)^{\!-m} 	.
	\end{eqnarray} 
	Inserting (\ref{MPG}) in (\ref{eq:kernal_A}) gives the expression in (\ref{eqn:kernal_A}).

	Given the Laplace transforms of $P_{I}$, by definition, the PDF of $P_{I}$ is attained by taking the inverse Laplace transform as follows:
	\begin{align}
	f_{P_{I}}(\rho) & =\mathcal{L}^{-1}\{\mathcal{L}_{P_{I}}(s)\}(\rho)  \nonumber \\ & =\mathcal{L}^{-1} \left\{ \mathrm{Det}\big(\mathrm{Id}+\alpha \mathbb{A}_{\Phi}(s)\big)^{-\frac{1}{\alpha}} \right \}(\rho), \label{eq:PDF_P_I}
	\end{align}
	with $\mathbb{A}_{\Phi}(s)$ given in (\ref{eqn:kernal_A}).
	
	Furthermore, integrating PDF in (\ref{eq:PDF_P_I}) yields
	\begin{align} \label{eq:CDF_PI}
	F_{P_{I}}(\rho) & \! = \!  \int^{\rho}_{-\infty}\!\!  \mathcal{L}^{-1} \! \left\{\mathcal{L}_{P_{I}}(s)\right \}(t) \mathrm{d}t \! = \! \mathcal{L}^{-\!1}   \left\{  \frac{\mathcal{L}_{P_{I}}(s)}{s} \right \}(\rho) \nonumber \\ &   = \mathcal{L}^{-1}  \left \{   \frac{ \mathrm{Det}\big(\mathrm{Id}+\alpha \mathbb{A}_{\Phi}(s)\big)^{-\frac{1}{\alpha} }}{s} \right \}(\rho)	.
	\end{align}
	
	When the hybrid transmitter is working in ambient backscattering mode, one can obtain the Laplace transform $\mathcal{L}_{P^{\mathrm{B}}_{E}}(s)$ 
	as 
	\begin{eqnarray} \label{eqn:LP_backscattering}
	\mathcal{L}_{P^{\mathrm{B}}_{E}}(s) = \mathbb{E} \left[ \exp\left(-s \beta \varrho P_{I} \right) \right]= \mathcal{L}_{P_{I}}(s\beta \varrho) 	.
	\end{eqnarray}

	Consequently, we can obtain the energy outage probability in ambient backscattering mode $\mathcal{O}_{\mathrm{B}}$, or equivalently, the CDF of $P^{\mathrm{B}}_{E}$ evaluated at $\rho_{\mathrm{B}}$, by integrating the PDF obtained in (\ref{eqn:LP_backscattering}) as 
	\begin{align} 
	\mathcal{O}_{\mathrm{B}} & =F_{P^{\mathrm{B}}_{E}} (\rho_{\mathrm{B}}) =F_{P_{I}} \left( \frac{\rho_{\mathrm{B}}}{\beta \varrho}\right) 
	= \mathcal{L}^{-1} \left\{ \frac{ \mathrm{Det}\big(\mathrm{Id}+\alpha \mathbb{A}_{\Phi}(s)\big)^{-\frac{1}{\alpha} } }{s } \right\}\left(\frac{\rho_{\mathrm{B}}}{\beta \varrho}\right) 	.	\label{CDF} 	
	\end{align} 
	
	Similarly, one obtains the energy outage probability in HTT mode $\mathcal{O}_{\mathrm{H}}$, or equivalently the CDF of $P^{\mathrm{H}}_{E}$ 
	evaluated at $\rho_{\mathrm{H}}$, as
	\begin{align} \label{eqn:CDF_HTT}
	\mathcal{O}_{\mathrm{H}} & =F_{P^{\mathrm{H}}_{E}}(\rho_{\mathrm{H}}) = F_{P_{I}} \left( \frac{\rho_{\mathrm{H}}}{\omega \beta  }\right) 
	= \mathcal{L}^{-1} \left\{ \frac{ \mathrm{Det}\big(\mathrm{Id}+\alpha \mathbb{A}_{\Phi}(s )\big)^{-\frac{1}{\alpha} }}{s } \right\} \left(\frac{\rho_{\mathrm{H}}}{\omega \beta}\right) 	.
	\end{align} 
	
	
Let $\mathcal{B}_{\mathrm{PTP}}$ 
denote the probability that the hybrid transmitter operated by PTP is in ambient backscattering mode. 
According to the criteria of PTP, from the definition in~(\ref{def:energyoutage}), we have
 \begin{align}
 \mathcal{O}_{\mathrm{PTP}} & = \mathcal{B}_{\mathrm{PTP}} \mathcal{O}_{\mathrm{B}}+ (1 - \mathcal{B}_{\mathrm{PTP}}) \mathcal{O}_{\mathrm{H}} \label{def:O_PTP}  
 \\ &=\mathbb{P} \left[P_{I} \leq \frac{\rho_{\mathrm{H}}}{\omega \beta} \right] F_{P_{I}} \left( \frac{\rho_{\mathrm{B}}}{\beta \varrho}\right)  
	+ \left(1-\mathbb{P} \left[P_{I} \leq \frac{\rho_{\mathrm{H}}}{\omega \beta}\right]\right) F_{P_{I}} \left( \frac{\rho_{\mathrm{H}}}{\omega \beta }\right) 
	.	\label{eqn:O_PTP} 
	\end{align} 
	
	One notices that $\mathcal{B}_{\mathrm{PTP}} $ is equal to the CDF of  $P_{I}$ evaluated at $\frac{\rho_{\mathrm{H}}}{\omega \beta}$, which is expressed as 
	\begin{align}\label{eqn:B_PTP}
 \mathcal{B}_{\mathrm{PTP}} 	& =   F_{P_{I}} \left( \frac{\rho_{\mathrm{H}}}{\omega \beta }\right)  
  = \mathcal{L}^{-1} \left\{ \frac{ \mathrm{Det}\big(\mathrm{Id}+\alpha \mathbb{A}_{\Phi}(s  )\big)^{-\frac{1}{\alpha} } }{s } \right\}\left(\frac{\rho_{\mathrm{H}}}{\omega \beta}\right) 	.
	\end{align}
	
	
	Then, by inserting (\ref{CDF}), (\ref{eqn:CDF_HTT}) and (\ref{eqn:B_PTP}) in (\ref{eqn:O_PTP}), we obtain $\mathcal{O}_{\mathrm{PTP}}$
	in (\ref{eq:energy_outage_probability_PTP}).	
\end{proof} 


\section*{Appendix II}

\begin{proof}
According to the criteria of STP, we have the probability of being in ambient backscattering mode as
\begin{align}
\mathcal{B}_{\mathrm{STP}} & \triangleq \mathbb{P} [\nu_{\mathrm{B}} > \tau_{\mathrm{B}}, P^{\mathrm{B}}_{E} >\rho_{\mathrm{B}}] \label{def:STP}  \\ 
& = \mathbb{P} \left [ \frac{\delta P_{I} h_{\mathrm{S},\mathrm{D}} }{d^{\mu} \sigma^2} \left (1- \varrho \right ) > \tau_{\mathrm{B}}, P^{\mathrm{B}}_{E} >\rho_{\mathrm{B}} \right] 
= \mathbb{P} \left [ h_{\mathrm{S},\mathrm{D}}\! >\!\frac{ \tau_{\mathrm{B}}d^{\mu} \sigma^2 }{\delta P_{I} \left ( 1- \varrho \right ) }, P_{I} \beta \varrho > \rho_{\mathrm{B}} \right] \nonumber \\ 
& \overset{\text{(a)}}{=}  \mathbb{P}   \left [ h_{\mathrm{S},\mathrm{D}}\! >\!\frac{ \tau_{\mathrm{B}}d^{\mu} \sigma^2 }{\delta P_{I} \left ( 1 \! -\! \varrho \right ) }   \Big \vert  P_{I}  \! >  \! \frac{\rho_{\mathrm{B}}}{\beta\varrho} \right] \! \mathbb{P} \! \left [ P_{I}  \! >  \! \frac{\rho_{\mathrm{B}}}{\beta\varrho} \right] 
= \mathbb E_{P_{I}}\!\! \left[ \mathbb{P} \! \left[h_{\mathrm{S},\mathrm{D}}\! >\!\frac{ \tau_{\mathrm{B}}d^{\mu} \sigma^2 }{\delta P_{I} \left ( 1- \varrho \right ) } \Bigg|  P_{I} \right]\!\! \mathbbm{1}_{ \{P_{I} > \frac{\rho_{\mathrm{B}}}{\beta \varrho} \}} \! \right] \nonumber\\ 
& 
= \int^{\infty}_{\frac{\rho_{\mathrm{B}}}{\beta\varrho}} \exp \left( - \frac{\lambda \tau_{\mathrm{B}} d^{\mu} \sigma^2 }{\delta \rho \left ( 1- \varrho \right ) } \right) f_{P_{I}}(\rho) \mathrm{d}\rho,	\label{eqn:delta_B_STP} 
\end{align} 
where (a) follows by the Bayes' theorem~\cite[page 36]{B.2008Ash},  
and $\mathbbm{1}_{\{\mathrm{E}\}}$ is an indicator function that takes the value of 1 if event $\mathrm{E}$ happens, and takes the value of 0 otherwise. 

Then, by replacing $\mathcal{B}_{\mathrm{PTP}}$ in the expression of
(\ref{def:O_PTP}) with $\mathcal{B}_{\mathrm{STP}}$ shown as (\ref{eqn:delta_B_STP}), we have (\ref{eq:energy_outage_probability_STP}) in {\bf Theorem}~\ref{thm:PowerOutage_STP} after some mathematical manipulations. 
\end{proof} 
 
\section*{Appendix III}


\begin{proof}

When there exists no repulsion, the GPP becomes a PPP with $\alpha$ approaching zero. By using the expansion~\cite{ShiraiTakahashi}
\begin{eqnarray}
	\mathrm{Det}\big(\mathrm{Id}+\alpha \mathbb{A}_{\Phi}(s)\big)^{-\frac{1}{\alpha} }
	\stackrel{\alpha \to 0}{\longrightarrow} \exp \left( -\int_{\mathbb{O}_{\mathrm{S}}} \mathbb{A}_{\Phi} (\mathbf{x},\mathbf{x} )\mathrm{d} \mathbf{x}\right),
	\end{eqnarray}
	we can simplify (\ref{eq:PDF_P_I}) as follows when $h_{a,\mathrm{S}} \sim \mathcal{E}(1)$ and $\mu=4$. 
\begin{align} 
f_{P_{I}}(\rho) 
&\!= \!\mathcal{L}^{-1}\! \left\{ \!\exp\! \left (- 2 \pi \zeta_{A}   \int^{R\to\infty}_{0}   \frac{r}{ 1\! + \! r^{4} (s  P_{A})^{-1} } \mathrm{d}r \right)  \right\}(\rho) 
 =  \mathcal{L}^{-1}\! \left\{ \frac{1}{s} \exp \left (\!-  \frac{ \pi^2 \zeta_{A} \sqrt{s P_{A} } }{ 2 }  \!\right) \! \right\}(\rho) \nonumber \\
&\overset{\text{(ii)}}{=} \frac{1}{2\pi i} \lim_{T \to \infty} \int^{z+i T}_{z-i T} \exp \left (\rho s - \frac{ \pi^2 \zeta_{A} \sqrt{s P_{A} } }{ 2 } \right) \mathrm{d} s \nonumber \\
&\overset{\text{(iii)}}{=} \frac{1}{2\pi i} \int^{\infty}_{0} \exp(-\rho t) \bigg[ \exp\left( \frac{ \pi^2 \zeta_{A} \sqrt{-t P_{A} } }{ 2 } \right) 
 - \exp\left( -\frac{ \pi^2 \zeta_{A} \sqrt{-t P_{A} } }{ 2 } \right) \bigg] \mathrm{d} t \nonumber \\ 
& \overset{\text{(iv)}}{=}\frac{1}{ \pi } \int^{\infty}_{0} \exp\left (-\rho \frac{4 u^2 }{\pi^4 \zeta^2_{A} P_{A}}\right) \sin(u) \frac{8u}{\pi^4 \zeta^2_{A} P_{A}} \mathrm{d}u
 \nonumber \\ & 
\overset{\text{(v)}}{=} \frac{ 1 }{4} \left(\frac{\pi}{\rho}\right)^{\!\frac{3}{2}} \zeta_{A} \sqrt{P_{A}} \exp \left(- \frac{\pi^4 \zeta_{A}^2 P_{A}}{16 \rho} \right), \nonumber
\end{align} 
where $(\text{ii})$ follows Mellin's inverse formula~\cite{P.1995Flajolet} which transforms the inverse Laplace transform into the complex plane, $i$ is the imaginary unit, i.e., $i\!=\!\sqrt{-1}$, and $z$ is a fixed constant greater than the real parts of the singularities of  $\exp \left ( - \frac{ \pi^2 \zeta_{A} \sqrt{s P_{A} } }{ 2 } \right)$, (\text{iii}) applies the Bromwich inversion theorem with the modified contour ~\cite[Chapter 2]{M.Cohen2007}, (\text{iv}) applies Euler's formula~\cite[Page 1035]{Miller2009} and a replacement of $u=\frac{ \pi^2 \zeta_{A} \sqrt{t P_{A}}}{2}$, and (\text{v}) uses the method of integration by parts.

Furthermore, based on the $f_{P_{I}}(\rho)$ expression, the CDF $F_{P_{I}}(\rho)$ in (\ref{corollary:PDF_no_replusion_Rayleigh}) can be obtained after some mathematical manipulations.
\end{proof} 
\section*{Appendix IV}
\label{proof_theorem2}
\begin{proof}
	We first determine the coverage probability in ambient backscattering mode. One simply notes that the expression of $\mathcal{C}_{\mathrm{B}}$ in (\ref{def:coverage_probability}) is equivalent to the definition of $\mathcal{B}_{\mathrm{STP}} $ in (\ref{def:STP}).  
	Hence, we have
	\begin{align}
	\mathcal{C}_{\mathrm{B}} = \int^{\infty}_{\frac{\rho_{\mathrm{B}}}{\beta \varrho}} \exp \left( - \frac{\lambda \tau_{\mathrm{B}}d^{\mu} \sigma^2 }{\delta \rho \left ( 1- \varrho \right ) } \right) f_{P_{I}}(\rho) \mathrm{d}\rho 	.	\label{eqn:delta_B_PTP} 
	\end{align} 
	
	Let $Q=\xi \sum_{b \in \mathcal{B}} P_{B} \widetilde{h}_{b,\mathrm{D}} \|\mathbf{x}_{b}-\mathbf{x}_{\mathrm{D}}\|^{-\mu}$ denote the aggregated interference at the receiver. We then derive the coverage probability in HTT mode as
	\begin{align}
	\mathcal{C}_{\mathrm{H}} 
	&= \mathbb P [ \nu_{\mathrm{H}} > \tau_{\mathrm{H}}, P^{\mathrm{H}}_{E} > \rho_{\mathrm{H}}]  
	\nonumber \\ & = \mathbb E_{P_{I}} \!   \left[ \mathbb P   \left[ \widetilde{h}_{\mathrm{S},\mathrm{D}}  \! >\! \frac{\tau_{\mathrm{H}}d^{ \mu} (1-\omega)( Q +\sigma^{2}) }{\omega \beta P_{I} - \rho_{\mathrm{H}} } \bigg \lvert P_{I}  \right] \!  \mathbbm{1}_{ \{P_{I} > \frac{\rho_{\mathrm{H}} }{\beta\omega} \} }\! \right] \nonumber \\ 
	&=  \mathbb E_{P_{I}}   \!  \Bigg[  \exp\! \left(\! - \frac{\lambda \tau_{\mathrm{H}}d^{\mu} \sigma^{2} (1\!-\!\omega) }{\omega \beta P_{I} - \rho_{\mathrm{H}} }  \!\right) \!  \mathbb{E}   \Bigg[\! \exp \! \bigg( \! \! - \frac{\lambda \tau_{\mathrm{H}}d^{\mu} (1\!-\!\omega) }{\omega \beta P_{I} - \rho_{\mathrm{H}} } 
	\xi \sum_{b \in \mathcal{B}}\! P_{B} \widetilde{h}_{b,\mathrm{D}} \|\mathbf{x}_{b}-\mathbf{x}_{\mathrm{D}}\|^{-\mu} \! \bigg) \!\Bigg]  \mathbbm{1}_{\{P_{I} > \frac{\rho_{\mathrm{H}} }{\beta\omega} \}} \!\Bigg]\nonumber \\ 
	& \overset{\text{(vi)}}{=} \int^{\infty}_{\frac{\rho_{\mathrm{H}}}{\beta\omega} } \exp \left( - \frac{\lambda \tau_{\mathrm{H}}d^{\mu} \sigma^{2} (1-\omega) }{\omega \beta \rho - \rho_{\mathrm{H}} } \right) 
	\mathrm{Det}\big( \mathrm{Id} + \alpha \mathbb{B}_{\Psi} (\rho) \big)^{-\frac{1}{\alpha}} f_{P_{I}}(\rho) \mathrm{d}\rho, \hspace{-2mm} \label{eqn:coverageprobability_RSP_HTT} 
	\end{align}
	where ($\text{vi}$) is given following Proposition~\ref{lemma1}, and $\mathbb{B}_{\Psi}(\rho)$ is defined in 
	(\ref{eqn:kernel_B}).
	
	By definition in (\ref{def:coverage_probability}), the coverage probability under PTP can be written as 
	\begin{align}
	\mathcal{C}_{\mathrm{PTP}} = \mathcal{B}_{\mathrm{PTP}} \mathcal{C}_{\mathrm{B}}+ (1-\mathcal{B}_{\mathrm{PTP}})\mathcal{C}_{\mathrm{H}} 	.	\label{eqn:overall_coverage_probability}
	\end{align}	
	Then, by plugging $\mathcal{B}_{\mathrm{PTP}}$ shown as (\ref{eqn:B_PTP}), $\mathcal{C}_{\mathrm{B}}$ shown as (\ref{eqn:delta_B_PTP}) and $\mathcal{C}_{\mathrm{H}} $ shown as (\ref{eqn:coverageprobability_RSP_HTT}) into (\ref{eqn:overall_coverage_probability}), we have 
	(\ref{eq:coverage_probability_Rayleigh_PTP}).
\end{proof} 

\section*{Appendix V}
\begin{proof} 
	The average throughput in ambient backscattering mode
	$\mathcal{T}_{\mathrm{B}}$ can be calculated as
	\begin{align}
	\mathcal{T}_{\mathrm{B}}  & = \mathbb{E} [ T_{\mathrm{B}}  \mathbbm{1}_{\{\nu_{\mathrm{B}} > \tau_{\mathrm{B}}, P^{\mathrm{B}}_{E} > \rho_{\mathrm{B}}\} } ]  
    = T_{\mathrm{B}} \mathbb{P} [\nu_{\mathrm{B}} > \tau_{\mathrm{B}}, P^{\mathrm{B}}_{E} > \rho_{\mathrm{B}}] = T_{\mathrm{B}} \mathcal{C}_{\mathrm{B}}, \label{eqn:throughput_B} 
	\end{align}
	where $T_{\mathrm{B}}$ has been defined in Subsection~\ref{sec:network_model} and $\mathcal{C}_{\mathrm{B}}$ has been obtained in (\ref{eqn:delta_B_PTP}). 
	
	Moreover, the average throughput in HTT mode can be computed as 
	\begin{align}
	&  \mathcal{T}_{\mathrm{H}}  =  \mathbb{E} [ ( 1 - \omega ) W  \log_{2} (1+\nu_{\mathrm{H}})\mathbbm{1}_{\{\nu_{\mathrm{H}} > \tau_{\mathrm{H}}, P^{\mathrm{H}}_{E} > \rho_{\mathrm{H}}\} } ]
	\nonumber \\ &
	\overset{\text{(vii)}}{=} \!\!(1\!-\!\omega) W \mathbb{E} \! \left[\int^{\infty}_{0} \! \! \mathbb{P}[ \log_{2}(1+\nu_{\mathrm{H}})\!>\! t] \mathrm{d}t \mathbbm{1}_{ \{\nu_{\mathrm{H}} > \tau_{\mathrm{H}}, P^{\mathrm{H}}_{E} > \rho_{\mathrm{H}}\} } \! \right] \nonumber \\ & 
	= \!(1\!-\!\omega) W \! \int^{\infty}_{\log_{2} (1+\tau_{\mathrm{H}})} \! \mathbb{E}_{P_{I}} \bigg[  \exp \left( \! - \frac{ \lambda d^{ \mu} (1 -\omega) (2^t - 1 ) }{ \omega \beta P_{I}  -  \rho_{\mathrm{H}} }  \right)\! 
	\bigg( \sigma^2 +  \xi \sum_{b \in \mathcal{B}}  P_{B} \widetilde{h}_{b,\mathrm{D}} \|\mathbf{x}_{b}-\mathbf{x}_{\mathrm{D}}\|^{-\mu}  \bigg)   \mathbbm{1}_{\{P_{I}>\frac{\rho_{\mathrm{H}}}{ \beta\omega}\}} \! \bigg]  \mathrm{d}t  \nonumber \\
	& \overset{\text{(viii)}}{=} \! (1 \!-\!\omega) W \! \! \int^{\infty}_{ \log_{2} (1+\tau_{\mathrm{H}} )}\!  \int^{\infty}_{\frac{\rho_{\mathrm{H}}}{\beta\omega}}    \exp \! \left( \!-  \frac{ \lambda d^{ \mu}\sigma^{2} (1\!-\!\omega) (2^t\!-\!1) }{ \omega \beta \rho - \rho_{\mathrm{H}} }  \!\right) 
 	\mathrm{Det}\big(\mathrm{Id} \!+ \!\alpha \mathbb{C}_{\Psi}(\rho)\big)^{-\frac{1}{\alpha}}   f_{P_{I}}(\rho) \mathrm{d}\rho \mathrm{d}t, \hspace{-2mm}	\label{eqn:throughput_HTT} 
	\end{align}
	where $\text{(vii)}$ follows $\mathbb{E}[X]= \int^{\infty}_{0}\mathbb{P}[X>x]\mathrm{d}x$~\cite{G.2011Andrews}, $\text{(viii)}$ is derived by applying Proposition~\ref{lemma1}, and $\mathbb{C}_{\Psi} (\rho)$ is defined in (\ref{eqn:kernel_C}).

	
	By definition in (\ref{def:throughput}), the average throughput under PTP can be written as 
	\begin{align}
	\mathcal{T}_{\mathrm{PTP}} 
	= \mathcal{B}_{\mathrm{PTP}} T_{\mathrm{B}} \mathcal{C}_{\mathrm{B}} + (1-\mathcal{B}_{\mathrm{PTP}}) \mathcal{T}_{\mathrm{H}} 	.	\label{eqn:average_throughput_PTP}
	\end{align} 
Inserting $\mathcal{B}_{\mathrm{PTP}}$ shown as  (\ref{eqn:B_PTP}), $\mathcal{C}_{\mathrm{B}}$ shown as (\ref{eqn:delta_B_PTP}), and $\mathcal{T}_{\mathrm{H}}$ shown as (\ref{eqn:throughput_HTT}) into (\ref{eqn:average_throughput_PTP}) yields (\ref{eq:throughput_PTP}).
\end{proof}



\end{document}